\pdfoutput=1
\documentclass[11pt]{article}

\usepackage{style}
\usepackage{shortcuts}

\newcommand*\samethanks[1][\value{footnote}]{\footnotemark[#1]}
\begin{document}

\title{Faster Knapsack Algorithms via Bounded Monotone Min-Plus-Convolution\thanks{This work is part of the project TIPEA that has received funding from the European Research Council (ERC) under the European Unions Horizon 2020 research and innovation programme (grant agreement No.~850979).}}

\author{Karl Bringmann\thanks{Saarland University and Max Planck Institute for Informatics, Saarland Informatics Campus, Saarbr\"ucken, Germany.} \and Alejandro Cassis\samethanks}

\maketitle

\begin{abstract}
        We present new exact and approximation algorithms for 0-1-Knapsack and Unbounded Knapsack:
        \begin{itemize}
            \item \emph{Exact Algorithm for 0-1-Knapsack:} 0-1-Knapsack has known algorithms running in time $\widetilde{O}(n + \min\{n \cdot \mathrm{OPT}, n \cdot W, \mathrm{OPT}^2, W^2\})$ [Bellman '57], where $n$ is the number of items, $W$ is the weight budget, and $\mathrm{OPT}$ is the optimal profit. We present an algorithm running in time $\widetilde{O}(n + (W + \mathrm{OPT})^{1.5})$. This improves the running time in case $n,W,\mathrm{OPT}$ are roughly equal. 
            
            \item \emph{Exact Algorithm for Unbounded Knapsack:} Unbounded Knapsack has known algorithms running in time $\widetilde{O}(n + \min\{n \cdot p_{\max}, n \cdot w_{\max}, p_{\max}^2, w_{\max}^2\})$ [Axiotis, Tzamos '19, Jansen, Rohwedder '19, Chan, He '22], where $n$ is the number of items, $w_{\max}$ is the largest weight of any item, and $p_{\max}$ is the largest profit of any item. We present an algorithm running in time $\widetilde{O}(n + (p_{\max} + w_{\max})^{1.5})$, giving a similar improvement as for 0-1-Knapsack.
        
            \item \emph{Approximating Unbounded Knapsack with Resource Augmentation:} Unbounded Knapsack has a known FPTAS with running time $\widetilde{O}(\min\{n/\varepsilon, n + 1/\varepsilon^2\})$ [Jansen, Kraft '18]. We study \emph{weak} approximation algorithms, which approximate the optimal profit but are allowed to overshoot the weight constraint (i.e. resource augmentation). We present the first approximation scheme for Unbounded Knapsack in this setting, achieving running time $\widetilde{O}(n + 1/\varepsilon^{1.5})$. Along the way, we also give a simpler FPTAS with lower order improvement in the standard setting.
        \end{itemize}
        
        For all of these problem settings the previously known results had matching conditional lower bounds. We avoid these lower bounds in the approximation setting by allowing resource augmentation, and in the exact setting by analyzing the time complexity in terms of weight and profit parameters (instead of only weight or only profit parameters).
        
        Our algorithms can be seen as reductions to Min-Plus-Convolution on monotone sequences with bounded entries. These structured instances of Min-Plus-Convolution can be solved in time $O(n^{1.5})$ [Chi, Duan, Xie, Zhang '22] (in contrast to the conjectured $n^{2-o(1)}$ lower bound for the general case). We complement our results by showing reductions in the opposite direction, that is, we show that achieving our results with the constant 1.5 replaced by any constant $<2$ implies subquadratic algorithms for Min-Plus-Convolution on monotone sequences with bounded entries.
\end{abstract}

\newpage


\section{Introduction}

In this paper we present new exact and approximation algorithms for Knapsack problems.

\paragraph*{Exact Pseudopolynomial Algorithms for \UnboundedKnapsack{}}
In the \styleprob{Unbounded} \linebreak \styleprob{Knapsack} problem we are given a set of $n$ items, 
where each item has a weight $w_i$ and a profit $p_i$, along with a knapsack capacity $W$. 
The goal is to find a \emph{multiset} of items which maximizes the total profit and has total weight at most $W$.
The textbook dynamic programming algorithm for \UnboundedKnapsack{} due to Bellman~\cite{Bellman57} runs in time $O(n W)$ 
or in time $O(n \cdot\OPT)$, where $\OPT$ is the value of the optimal solution. 
Recent literature on \UnboundedKnapsack{} studies alternative parameters: Axiotis and Tzamos~\cite{AxiotisT19} and 
Jansen and Rohwedder~\cite{JansenR19} independently presented algorithms running in time $\widetilde{O}(w_{\max}^2)$ 
and $\widetilde{O}(p_{\max}^2)$,\footnote{We use the notation $\widetilde{O}(T) = \cup_{c > 0} O(T \log^c T)$ to 
supress polylogarithmic factors} where $w_{\max}$ is the largest weight of any item and $p_{\max}$ the largest profit 
of any item --- in general $w_{\max}$ could be much smaller than $W$ and $p_{\max}$ much smaller than $\OPT$, so these algorithms improve upon 
Bellman's algorithm for some parameter settings. Chan and He~\cite{ChanH22} presented further 
improvements\footnote{Note that we can assume that $n \leq w_{\max}$ without loss of generality since if there are multiple items with 
the same weight, we can keep only the one with the largest profit. Similarly, $n \leq p_{\max}$.} 
achieving time $\widetilde{O}(n \, w_{\max})$ and $\widetilde{O}(n \, p_{\max})$
Note that when $w_{\max} \approx p_{\max} \approx n$ all mentioned algorithms require at least quadratic time 
$\Omega(n^2)$. Can we overcome this quadratic barrier? In this paper we answer this positively by considering the 
combined parameter $w_{\max} + p_{\max}$. 

\begin{theorem}\label{thm:alg-unboundedknapsack}
    \UnboundedKnapsack{} can be solved in expected time $\widetilde{O}(n + (p_{\max} + w_{\max})^{1.5})$.
\end{theorem}

This result is particularly interesting in light of recent fine-grained lower bounds for \UnboundedKnapsack{}. 
Indeed, for each previous result that we have mentioned above, a matching conditional lower bound is 
known~\cite{CyganMWW19, KunnemannPS17}. For example, \UnboundedKnapsack{} cannot be solved in time $O((n+W)^{2-\delta})$ 
for any constant $\delta > 0$ under a plausible hypothesis. Inspecting these conditional lower bounds, we observe 
that they construct hard instances where only the profit parameters or only the weight parameters are under control; 
one of the two must be very large to obtain a hardness reduction. We thus avoid these conditional lower bounds by 
considering the combined profit and weight parameter $w_{\max} + p_{\max}$.

\paragraph*{Exact Pseudopolynomial Algorithms for \Knapsack{}}

The \Knapsack{} problem is the variant of \UnboundedKnapsack{} where every input item can appear at most once in any solution. 
Bellman's algorithm also solves \Knapsack{} in time $O(n W)$ or $O(n\cdot \OPT)$. However, the landscape is more diverse when 
considering other parameters. In particular, it is open whether \Knapsack{} can be solved in time $\widetilde{O}(n + w_{\max}^2)$ 
or $\widetilde{O}(n + p_{\max}^2)$.\footnote{This gap is analogous to the case of \SubsetSum{} where we are given a set 
of numbers $X$ and a target number $t$. For the unbounded case, where the goal is to find whether a multiset of items in $X$ sums to $t$,
Jansen and Rohwedder~\cite{JansenR19} gave an algorithm in time $\widetilde{O}(n + u)$ where $u$ is the largest number in the input.
For the more standard ``0-1'' case where we ask for a subset of $X$ summing to $t$, the best known running times are 
$\widetilde{O}(n + t)$~\cite{Bringmann17,JinW19}, $O(nu)$~\cite{Pisinger99}, $\widetilde{O}(n + u^2/n)$ by combining~\cite{GalilM91} 
and~\cite{Bringmann17,JinW19}, and $\widetilde{O}(n + u^{3/2})$ by combining \cite{GalilM91} and \cite{Bringmann17,JinW19} and \cite{Pisinger99};
see also~\cite{BringmannW21, PRW21} for generalizations to $X$ being a multiset and related results.} 
\Cref{tab:running-times} shows a non-exhaustive list of 
pseudopolynomial-time algorithms for \Knapsack{} using different combinations of the parameters $n,w_{\max}, p_{\max}, W, \OPT$. 
Note that when all these parameters are bounded by $O(n)$, all existing algorithms require at least quadratic time~$\Omega(n^2)$. 
In this paper we show that by considering the combined weight and profit parameter $W + \OPT$ we can overcome this quadratic barrier.

\begin{theorem}\label{thm:alg-knapsack}
    There is a randomized algorithm for \Knapsack{} that runs in time $\widetilde{O}(n + (W + \OPT)^{1.5})$ and 
    succeeds with high probability.
\end{theorem}

Similar to the unbounded case, matching conditional lower bounds ruling out time $O((n+W)^{2-\delta})$ and 
$O((n+\OPT)^{2-\delta})$ for any $\delta > 0$ are known~\cite{CyganMWW19, KunnemannPS17}. These lower bounds 
construct hard instances where only one of $W,\OPT$ is under control, the other needs to be very large. 
We thus avoid these lower bounds by considering the combined weight and profit parameter $W + \OPT$.

\bgroup
\def\arraystretch{1.2}
\begin{table}[t]
    \caption{Non-exhaustive list of pseudopolynomial-time algorithms for \Knapsack{}.} \label{tab:running-times}
    \centering
    \begin{tabular}{|ll|}
        \hline
        \thead{Running Time} & \thead{Reference} \\
        \hline
        \makecell{$O(n \cdot \min\{W, \OPT\})$} & \cite{Bellman57} \\
        \makecell{$O(n \cdot p_{\max} \cdot w_{\max})$} & \cite{Pisinger99} \\
        \makecell{$\widetilde{O}(n + w_{\max} \cdot W)$} & \cite{KellererP04,BateniHSS18,AxiotisT19} \\
        \makecell{$\widetilde{O}(n + p_{\max} \cdot W)$} & \cite{BateniHSS18} \\
        \makecell{$\widetilde{O}(n \cdot \min\{w_{\max}^2, p_{\max}^2\})$} & \cite{AxiotisT19} \\
        \makecell{$\widetilde{O}((n + W) \cdot \min\{w_{\max}, p_{\max}\})$} & \cite{BateniHSS18} \\
        \makecell{$O(n + \min\{w_{\max}^3, p_{\max}^3\})$} & \cite{PRW21} \\
        \makecell{$\widetilde{O}(n + (W + \OPT)^{1.5})$} & \textbf{\Cref{thm:alg-knapsack}} \\
        \hline
    \end{tabular}
\end{table}
\egroup

\paragraph*{Approximation Schemes for \UnboundedKnapsack{}}

Since \UnboundedKnapsack{} is well known to be NP-hard, it is natural to study approximation algorithms. 
In particular, a \emph{fully polynomial-time approximation scheme} (FPTAS) given $0 < \varepsilon < 1$ computes a 
solution $x$ with total weight $w(x) \leq W$ and total profit $p(x) \geq (1-\varepsilon)\OPT$ in time $\poly(n, 1/\varepsilon)$.
The first FPTAS for \UnboundedKnapsack{} was designed by Ibarra and Kim in 1975~\cite{IbarraK75} and runs 
in time $\widetilde{O}(n + (1/\varepsilon)^4)$. In 1979 Lawler~\cite{Lawler79} improved the running time to 
$\widetilde{O}(n + (1/\varepsilon)^3)$. This was the best known until Jansen and Kraft in 2018~\cite{JansenK18} 
presented an FPTAS running in time $\widetilde{O}(n + (1/\varepsilon)^2)$. This algorithm has a matching conditional 
lower ruling out time $O((n + 1/\varepsilon)^{2-\delta})$ for any $\delta > 0$~\cite{CyganMWW19, KunnemannPS17,MuchaW019}.

We present a new FPTAS for \UnboundedKnapsack{} which is (as we believe) simpler than Jansen and Kraft's, and 
has a lower order improvement in the running time:

\begin{theorem}\label{thm:fptas}
    \UnboundedKnapsack{} has an FPTAS with running time $\widetilde{O}\bigg(n + \frac{(1/\varepsilon)^2}{2^{\Omega(\sqrt{\log(1/\varepsilon)})}}\bigg)$.
\end{theorem}

Bringmann and Nakos~\cite{BringmannN21} recently gave an FPTAS for the related \SubsetSum{} problem which achieves the same running time.

\paragraph*{Weak Approximation for Unbounded Knapsack}

Motivated by the matching upper and conditional lower bounds for FPTASs for \UnboundedKnapsack{}, we
study the relaxed notion of \emph{weak approximation} as coined in~\cite{MuchaW019}: we relax the weight constraint
and seek a solution $x$ with total weight $w(x) \leq (1+\varepsilon) W$ and total profit
$p(x) \geq (1-\varepsilon)\OPT$. Note that $\OPT$ is still the optimal value of any solution with weight at most $W$.
This can be interpreted as bicriteria approximation (approximating both the weight and profit constraint) 
or as resource augmentation (the optimal algorithm is allowed weight $W$ while our algorithm is allowed a slightly 
larger weight of $(1+\varepsilon)W$). All of these are well-established relaxations of the standard 
(=strong\footnote{From now on, by ``strong'' approximation we mean the standard (non-weak) notion of approximation.}) 
notion of approximation. Such weaker notions of approximation are typically studied
when a PTAS for the strong notion of approximation is not known. More generally, studying these weaker notions is justified
whenever there are certain limits for strong approximations, to see whether these limits can be overcome by relaxing the notion of approximation.
In particular, we want to understand whether this relaxation can overcome the conditional lower bound 
ruling out time $O((n + 1/\varepsilon)^{2-\delta})$. For the related \SubsetSum{} problem this question has been
resolved positively: Bringmann and Nakos~\cite{BringmannN21} conditionally ruled out strong approximation algorithms in time 
$O((n + 1/\varepsilon)^{2-\delta})$ for any $\delta > 0$, but Mucha, W\k{e}grzycki and W\l{}odarczyk~\cite{MuchaW019} 
designed a weak FPTAS in time $\widetilde{O}(n + (1/\varepsilon)^{5/3})$.

In this paper, we give a positive answer for \UnboundedKnapsack{}:

\begin{theorem}\label{thm:weak-fptas}
    \UnboundedKnapsack{} has a weak approximation scheme running in expected time 
    $
        \widetilde{O}(n + (\tfrac{1}{\varepsilon})^{1.5})
    $.
\end{theorem}

Our theorem gives reason to believe that resource augmentation indeed makes the problem easier. Specifically, 
obtaining a strong approximation scheme with the same running time as our weak one would refute 
the existing conditional lower bound for strong approximation.

\paragraph*{Min-Plus-Convolution}

All conditional lower bounds mentioned above are based on a hypothesis about the \MinConv{} problem:
Given sequences $A, B \in \Z^n$ compute their $(\min,+)$-convolution,
which is the sequence $C \in \Z^{2n}$ with $C[k] = \min_{i+j=k} A[i] + B[j]$.\footnote{If we replace the $\min$ by a $\max$, 
we obtain the \MaxConv{} problem, which is equivalent to \MinConv{} after negating the sequences. Therefore, we will 
sometimes use these two names interchangeably.}
The \MinConv{} problem can be trivially solved in time $O(n^2)$. This can be improved to time 
$n^2/2^{\Omega(\sqrt{\log n})}$ via a reduction to $(\min,+)$-matrix product due to Bremner et al.~\cite{BremnerCDEHILPT14},
and using Williams' algorithm for the latter~\cite{Williams18} (which was derandomized later by Chan and Williams~\cite{ChanW21}).
The lack of truly subquadratic algorithms despite considerable effort has led researchers to postulate the \MinConv{} hypothesis, 
namely that \MinConv{} cannot be solved in time $O(n^{2-\delta})$ for any constant $\delta > 0$~\cite{CyganMWW19,KunnemannPS17}.
Many problems are known to have conditional lower bounds from the \MinConv{}
hypothesis, see, e.g.,~\cite{BackursIS17,ChanH19,CyganMWW19,JansenR19,KunnemannPS17,LaberRC14}.

Central to our work is a reduction from \MinConv{} to \UnboundedKnapsack{} shown independently by Cygan et al.~\cite{CyganMWW19} 
and K\"{u}nnemann et al.~\cite{KunnemannPS17}. In particular, they showed that if \UnboundedKnapsack{} on $n$ items 
and $W = O(n)$ can be solved in subquadratic time, then \MinConv{} can be 
solved in subquadratic time. This reduction immediately implies matching conditional lower bounds for the previously known exact algorithms 
with running times $O(n W)$, $\widetilde{O}(w_{\max}^2)$ and $\widetilde{O}(n \, w_{\max})$, as mentioned earlier.

A small modification of this reduction extends to the dual case, i.e., an exact subquadratic-time algorithm for \UnboundedKnapsack{}
with $\OPT = O(n)$ would result in a subquadratic-time algorithm for \MinConv{}.
This establishes matching conditional lower bounds for the algorithms in time $O(n \cdot \OPT)$, 
$\widetilde{O}(p_{\max}^2)$ and $\widetilde{O}(n \, p_{\max})$. Moreover, by setting $\varepsilon = \Theta(1/\OPT)$,
an FPTAS for \UnboundedKnapsack{} would yield an exact algorithm for \MinConv{}, establishing that the $\widetilde{O}(n + (1/\varepsilon)^2)$-time
FPTAS is conditionally optimal. This last observation was pointed out in~\cite{MuchaW019}.

In fact, the reduction due to Cygan et al.~\cite{CyganMWW19} is even an equivalence, showing
that \UnboundedKnapsack{} is intimately connected to \MinConv{}. The same reduction is also known 
for \Knapsack{}~\cite{CyganMWW19,KunnemannPS17}, so a similar discussion applies to \Knapsack{}.

\paragraph*{Bounded Monotone Min-Plus-Convolution}

Despite its postulated hardness, there are restricted families of instances of \MinConv{} for which
subquadratic algorithms are known, see e.g.~\cite{BringmannGSW19,BussieckHWZ94,ChanL15,ChiDXZ22}. Central to this paper is the case 
when the input sequences $A, B \in \Z^n$ are monotone non-decreasing and have entries bounded by $O(n)$; we call this task \BMMinConv{}.

In a celebrated result, Chan and Lewenstein gave an algorithm for \BMMinConv{} that runs in expected time $O(n^{1.859})$~\cite{ChanL15}. 
As a big hammer, their algorithm uses the famous Balog-Szemer\'edi-Gowers theorem from additive combinatorics.
Very recently, Chi, Duan, Xie and Zhang showed how to avoid this big hammer and improved the running time to 
expected $\widetilde O(n^{1.5})$~\cite{ChiDXZ22}.

All of our results mentioned so far use this \BMMinConv{} algorithm as a subroutine.
That is, our algorithms are reductions from various Knapsack problems to \BMMinConv{}.

\paragraph*{Equivalence with Bounded Monotone Min-Plus-Convolution}

We complement our results by showing reductions in the opposite direction: Following the same 
chain of reductions as in~\cite{CyganMWW19,KunnemannPS17} but starting from bounded monotone instances of
$(\min,+)$-convolution, we reduce \BMMinConv{} to $O(n)$ instances of \UnboundedKnapsack{} with $O(\sqrt{n})$ 
items each, where it holds that $w_{\max},p_{\max},W$ and $\OPT$ are all bounded 
by $O(\sqrt{n})$. Instantiating this reduction for the exact and approximate setting, we show the following theorem.

\begin{restatable}[Equivalence]{theorem}{equivalence}\label{thm:equivalence}
    For any problems $A$ and $B$ from the following list, if $A$ can be solved in time $\widetilde{O}(n^{2-\delta})$
    for some $\delta > 0$, then $B$ can be solved in randomized time $\widetilde{O}(n^{2-\delta/2})$:
    \begin{enumerate}
        \item \BMMaxConv{} on sequences of length $n$
        \item \UnboundedKnapsack{} on $n$ items and $W, \OPT = O(n)$
        \item \Knapsack{} on $n$ items and $W, \OPT = O(n)$
        \item Weak $(1 + \varepsilon)$-approximation for \UnboundedKnapsack{} on $n$ items and $\varepsilon = \Theta(1/n)$
    \end{enumerate}
\end{restatable}

On the one hand, \Cref{thm:alg-unboundedknapsack} solves \UnboundedKnapsack{} in time 
$\widetilde{O}(n + (p_{\max} + w_{\max})^{1.5})$ by using Chi, Duan, Xie and Zhang's subquadratic \BMMinConv{} algorithm~\cite{ChiDXZ22}.
On the other hand, \Cref{thm:equivalence} shows that any algorithm solving \UnboundedKnapsack{} in time 
$\widetilde{O}(n + (p_{\max} + w_{\max})^{2-\delta})$ can be transformed into a subquadratic \BMMinConv{} algorithm. 
This shows that both our exact and approximation algorithms take essentially \emph{the only possible route} to obtain subquadratic 
algorithms, by invoking a \BMMinConv{} algorithm.

\paragraph*{Is randomness necessary?}
The algorithms given by Theorems~\ref{thm:alg-unboundedknapsack}, \ref{thm:alg-knapsack} and \ref{thm:weak-fptas} are all randomized.
If we insist on deterministic algorithms, we note that by applying Chan and Lewenstein's deterministic $\widetilde{O}(n^{1.864})$-time 
algorithm for \BMMaxConv{}~\cite{ChanL15}, we can obtain deterministic versions of~\Cref{thm:alg-unboundedknapsack} and \Cref{thm:weak-fptas}
with exponent $1.864$ instead of $1.5$ (i.e. the only part where we use randomness is in applying Chi, Duan, Xie and Zhang's 
algorithm~\cite{ChiDXZ22}). On the other hand, we do not know\footnote{Our algorithm closely 
follows Bringmann's algorithm for \SubsetSum{}~\cite{Bringmann17} whose derandomization is an open problem.} how to derandomize~\Cref{thm:alg-knapsack}.

\subsection{Technical Overview}

\paragraph*{Exact algorithm for \UnboundedKnapsack{}}
In~\Cref{sec:exact-unboundedknapsack} we present our exact algorithm for \UnboundedKnapsack{}.
Let $(\calI, W)$ be an instance of \UnboundedKnapsack{}.
We denote by $\calP_i[0,\dots,W]$ the array where $\calP_i[j]$ is the maximum profit of a solution of weight 
at most $j$ using at most $2^i$ items. Since each item has weight at least 1, any feasible 
solution consists of at most $W$ items. Thus, our goal is to compute the value $\calP_{\lceil \log W \rceil}[W] = \OPT$.

The natural approach is to use dynamic programming: since $\calP_0$ consists 
of solutions of at most one item, it can be computed in time $O(n)$. For 
$i > 0$ we can compute $\calP_i[0,\dots,W]$ by taking the $(\max, +)$-convolution of $\calP_{i-1}[0,\dots,W]$ with itself. 
This gives an algorithm in time $O(W^2 \log W)$.

Jansen and Rohwedder~\cite{JansenR19} and Axiotis and Tzamos~\cite{AxiotisT19} showed that instead of 
convolving sequences of length $W$, it suffices to convolve only $O(w_{\max})$ entries of $\calP_{i-1}$ in 
each iteration. This improves the running time to $O(w_{\max}^2 \log W)$ by using the naive algorithm for $(\max,+)$-convolution.
The approach of Jansen and Rohwedder~\cite{JansenR19} is as follows. Suppose $x$ is the optimal solution for 
a target value $\calP_i[j]$. They showed that $x$ can be split into two solutions $x_1, x_2$ such that (i) the number 
of items in each part is at most $2^{i-1}$ and (ii) the difference between the weights of both parts is at most $O(w_{\max})$. Thus,
(i) guarantees that both $x_1$ and $x_2$ are optimal solutions for two entries of $\calP_{i-1}$, and (ii)
implies that these entries lie in an interval in $\calP_{i-1}$ of length $O(w_{\max})$. In this way, they can afford 
to perform the $(\max,+)$-convolution of only $O(w_{\max})$ entries in $\calP_{i-1}$. 

To show the existence of such a partitioning of $x$ they make use of Steinitz' Lemma~\cite{Steinitz1913}, which shows that any collection of 
$m$ vectors in $\R^d$ with infinity norm at most 1, whose sum is 0, can be permuted such that every prefix sum has norm at most $O(d)$ 
(see~\Cref{lem:steinitz} for the precise statement). The partitioning of $x$ follows from Steinitz' Lemma by taking the weights of the items picked by $x$
as 1-dimensional vectors. The usage of Steinitz' Lemma to reduce the number of states in dynamic programs was pioneered by 
Eisenbrand and Weismantel~\cite{EisenbrandW20}, and further refined by Jansen and Rohwedder~\cite{JansenR19}.

In our algorithm, we use Steinitz' Lemma in a similar way to split the number of items and the weight of $x$, but additionally we use it to ensure that the profits of the solutions $x_1, x_2$ differ by at most $O(p_{\max})$ (see \Cref{lem:splitting-steinitz}). 
In this way, by carefully handling the subproblems $\calP_{i-1}$ we can enforce that the values of the $O(w_{\max})$ entries that 
need to be convolved have values in a range of size $O(p_{\max})$. Since the arrays $\calP_{i}$ are monotone non-decreasing, we 
can then apply the algorithm for \BMMaxConv{}, and thus handle each subproblem in time $\widetilde O((p_{\max} + w_{\max})^{1.5})$.

\paragraph*{Exact algorithm for \Knapsack{}}

Cygan et al.~\cite{CyganMWW19} showed that there is a reduction from \Knapsack{} to \MaxConv{}. More precisely,
they showed that if \MaxConv{} can be solved in time $T(n)$, then \Knapsack{} can be solved in randomized time 
$\widetilde{O}(T(W))$. Their reduction is a generalization of Bringmann's \SubsetSum{} algorithm~\cite{Bringmann17}, which can 
be seen as a reduction from \SubsetSum{} to Boolean convolution. Cygan et al. showed that the reduction for \Knapsack{}
can be obtained by essentially replacing the Boolean convolutions by $(\max,+)$-convolutions in Bringmann's algorithm.

In~\Cref{sec:exactzerooneknapsack} we observe that this reduction produces instances of \MaxConv{} which are monotone non-decreasing
and have entries bounded by $\OPT$. That is, we obtain \BMMaxConv{} instances of size $O(W+\OPT)$, and following the analysis of~\cite{CyganMWW19} we 
obtain an algorithm for \Knapsack{} in time $\widetilde{O}(n + (W + \OPT)^{1.5})$, which yields \Cref{thm:alg-knapsack}.

\paragraph*{Approximating \UnboundedKnapsack{}}

In~\Cref{sec:approx} we present our approximation schemes for \UnboundedKnapsack{}.
Let $(\calI, W)$ be an instance of \UnboundedKnapsack{}. The starting point is the ``repeated squaring'' approach mentioned earlier.
Namely, recall that we denote by $\calP_i[0,\dots,W]$
the sequence where $\calP_i[j]$ is the maximum profit of any solution of weight at most $j$ and with at most $2^i$ items.
Since $W/w_{\min}$ is an upper bound on the number of items in the optimal solution, we have that
$\calP_{\log(W/w_{\min})}[W] = \OPT$. We can compute $\calP_i[0,\dots,W]$
by taking the $(\max, +)$-convolution of $\calP_{i-1}[0,\dots,W]$ with itself. This yields an algorithm in time $O(W^2 \log W/w_{\min})$.

Although this exact algorithm is not particularly exciting or new, it can be nicely extended to the approximate setting.
In~\Cref{sec:fptas} we show that by replacing the exact $(\max,+)$-convolutions with approximate ones, 
we obtain an FPTAS for \UnboundedKnapsack{} in time $\widetilde{O}(n + 1/\varepsilon^2)$. To this end,
we preprocess the item set to get rid of \emph{light} items with weight smaller than $\varepsilon \cdot W$, and \emph{cheap}
items with profit smaller than $\varepsilon \cdot \OPT$, while decreasing the optimal value by only $O(\varepsilon \cdot \OPT)$ (see~\Cref{lem:preprocessing-profits-and-weights}). After this preprocessing, we have that the maximum number of items in any solution 
is $W/w_{\min} < 1/\varepsilon$. Thus, we only need to approximate $O(\log 1/\varepsilon)$ convolutions.
For this we use an algorithm due to Chan~\cite{Chan18a}, which in our setting without cheap items runs in time 
$\widetilde{O}(1/\varepsilon^2)$ (see \Cref{lem:max-conv-approx}). 
Thus, after applying the preprocessing in time $O(n)$, we compute a $(1 + \varepsilon)$-approximation of 
$\calP[0,\dots,W]$ by applying $O(\log 1/\varepsilon)$ approximate $(\max,+)$-convolutions
in overall time $\widetilde{O}(n + (1/\varepsilon)^2)$.

In~\Cref{sec:weak-fptas} we treat the case of weak approximation. The main steps of the algorithm are virtually the same as before. The crucial difference 
is that now we can afford to round weights. In this way, we can adapt Chan's algorithm and construct \MaxConv{}
instances which are monotone non-decreasing and have bounded entries (see~\Cref{lem:max-conv-weak-approx}). This yields \BMMaxConv{}
instances, and by using Chi, Duan, Xie and Zhang's algorithm for this special case~\cite{ChiDXZ22}, we can compute 
a weak approximation of $(\max,+)$-convolution in time $\widetilde{O}((1/\varepsilon)^{1.5})$. By similar arguments 
as for the strong approximation, this yields a weak approximation scheme running in time $\widetilde{O}(n + (1/\varepsilon)^{1.5})$.

\paragraph*{Equivalence between \BMMinConv{} and Knapsack problems}
As mentioned earlier in the introduction, Cygan et al.~\cite{CyganMWW19} and K\"{u}nnemann et al.~\cite{KunnemannPS17} 
independently showed a reduction from \MinConv{} to \UnboundedKnapsack{}.
In~\Cref{sec:hardness} we show that following the same chain of reductions from \MaxConv{} to \UnboundedKnapsack{}
but instead starting from \BMMaxConv{}, with minor adaptations we can produce instances of \UnboundedKnapsack{} with $W, \OPT = O(n)$. 
Together with our exact algorithm for \UnboundedKnapsack{}, which we can phrase as a reduction to \BMMaxConv{}, we 
obtain an equivalence of \BMMaxConv{} and \UnboundedKnapsack{} with $W, \OPT = O(n)$ --- if one of these problems
can be solved in subquadratic time, then both can.

Note that for \UnboundedKnapsack{} with $W, \OPT = O(n)$ a weak $(1+\varepsilon)$-approximation for
$\varepsilon = \Theta(1/n)$ already computes an exact optimal solution. This yields the reduction from \BMMaxConv{}
to the approximate version of \UnboundedKnapsack{}. We similarly obtain a reduction to \Knapsack{} with 
$W,\OPT = O(n)$. This yields our equivalences from \Cref{thm:equivalence}.

\section{Preliminaries}\label{sec:preliminaries}

We write $\N = \{0,1,2,\dots\}$. For $t \in \N$ we let $[t] = \{0,1,\dots,t\}$. 
For reals $a \le b$ we write $[a,b]$ for the interval from $a$ to $b$, and for reals $a,b$ with $b \ge 0$ 
we write $[a \pm b]$ for the interval $[a-b,a+b]$.

We formally define the \UnboundedKnapsack{} problem. We are given a set of items $\calI = \{(p_1, w_1),\dots,(p_n, w_n)\}$, 
where each item $i$ has a profit $p_i \in \N$ and a weight $w_i \in \N$, and a knapsack capacity $W \in \N$. 
The task is to maximize $\sum_{i=1}^n p_i \cdot x_i$ subject to the constraints $\sum_{i=1}^n w_i \cdot x_i \leq W$ and $x \in \N^n$. 
The more standard \Knapsack{} problem is defined in the same way, but the solution $x$ is constrained to $x \in \{0,1\}^n$. 

Given an instance $(\calI, W)$, we denote by $x \in \N^n$ a multiset of items, 
where $x_i$ is the number of copies of the $i$-th item. We sometimes refer to $x$ as a \emph{solution}. 
We write $p_{\calI}(x)$ for the total profit of $x$, i.e., $p_{\calI}(x) := \sum_i x_i \cdot p_i$. 
Similarly, we write $w_{\calI}(x) := \sum_i x_i \cdot w_i$ for the weight of $x$. 
When the item set $\calI$ is clear from context, we drop the subscript and simply write $p(x)$ and $w(x)$.
We denote the number of items contained in a solution $x$ by $\|x\|_1 := \sum_i x_i$.
A solution~$x$ is \emph{feasible} if it satisfies the constraint $w(x) \leq W$. 
We denote by $\OPT$ the maximum profit $p(x)$ of any feasible solution $x$.
We denote by $p_{\max} := \max_{(p, w) \in \calI} p$ the maximum profit of any input item and by
$w_{\max} := \max_{(p, w) \in \calI} w$ the maximum weight of any input item. 

\paragraph*{Notions of Approximation}

We say that an algorithm gives a \emph{strong} $(1+\varepsilon)$-approximation for \UnboundedKnapsack{}
if it returns a solution $x \in \N^n$ with weight $w(x) \leq W$ and profit
$p(x) \geq (1 - \varepsilon)\cdot \OPT$. We say that an algorithm gives a \emph{weak} $(1+\varepsilon)$-approximation 
for \UnboundedKnapsack{} if  it returns a solution $x \in \N^n$ with profit $p(x) \geq (1 - \varepsilon)\cdot  \OPT$
and weight $w(x) \leq (1 + \varepsilon) \cdot W$. 
We stress that here $\OPT$ still denotes the optimum value with weight at most $W$, i.e., $\OPT = \max\{p(x) \mid x \in \N^n,\, w(x) \le W\}$.

\paragraph*{Profit Sequences}
Given an item set $\calI$ and capacity $W$, we define the array $\calP_{\calI}[0,\dots,W]$, where $\calP_{\calI}[j]$ is the maximum profit achievable with capacity $j$, i.e.,
\[
    \calP_{\calI}[j] := \max\{p_{\calI}(x) \colon x \in \N^n, w_{\calI}(x) \leq j\}.    
\]
Note that $\calP_{\calI}[0] = 0$. A textbook way to compute $\calP_{\calI}[0,\dots,W]$ is by dynamic programming:

\begin{fact}\label{fact:dp}
    $\calP_{\calI}[0,\dots,W]$ can be computed using dynamic programming in time $O(n \cdot W)$.
\end{fact}

We will also consider the array $\calP_{\calI, k}[0,\dots,W]$, where we restrict to solutions with at most $2^k$ items, for any non-negative integer $k$, i.e., for any $j \in [W]$ we set
\[
    \calP_{\calI,k}[j] := \max\{p_{\calI}(x) \colon  x \in \N^n, w(x) \leq j, \|x\|_1 \leq 2^k\}.
\]

When $\calI$ is clear from context, we will drop the subscript and write $\calP[0,\dots,W]$ and $\calP_k[0,\dots,W]$.
When we work with \Knapsack{} instead of \UnboundedKnapsack{}, we will use the same notation $\calP_{\calI}$ and 
$\calP_{\calI, k}$, where we restrict to $x \in \{0,1\}^n$ instead of $x \in \N^n$. 

\paragraph*{\MaxConv{}} The $(\max,+)$-convolution $A \oplus B$ of two sequences $A[0,\dots,n], B[0,\dots,n] \in \Z^{n+1}$
is a sequence of length $2n+1$ where $(A \oplus B)[k] := \max_{i + j = k}A[i] + B[j]$. We call \MaxConv{} the 
task of computing the $(\max,+)$-convolution of two given sequences.

We will use the following handy notation: Given sequences $A[0,\dots,n], B[0,\dots,n]$ and 
intervals $I, J \subseteq [n]$ and $K \subseteq [2n]$, we denote by $C[K] := A[I] \oplus B[J]$ the computation of 
$C[k] := \max\{A[i] + B[j] \colon i \in I, j \in J, i + j = k\}$ for each $k \in K$. 
The following proposition shows that this can be computed efficiently. We defer the proof to~\autoref{sec:proofs-preliminaries}.

\begin{proposition}\label{prop:maxconv-range}
    If \MaxConv{} can be solved in time $T(n)$, then $C[K] = A[I] \oplus B[J]$ can be computed in time 
    $O(T(|I| + |J|) + |K|)$.
\end{proposition}

Sometimes we will refer to the $(\min,+)$-convolution, where we replace $\max$ by a $\min$. The two problems \MinConv{} and \MaxConv{}
are equivalent after negating the sequences.

\paragraph*{\BMMaxConv{}} In the \BMMaxConv{} problem, we compute the $(\max,+)$-convolution of sequences 
of length $n$ which are monotone non-decreasing and have bounded values.
For this setting, Chi, Duan, Xie and Zhang gave the following remarkable result:

\begin{theorem}[\BMMaxConv{}~\cite{ChiDXZ22}]\label{thm:convolution-bounded}
    Given monotone non-decreasing sequences $A[0,\dots,n]$ and $B[0,\dots,n]$ with 
    entries $A[i], B[i] \in [O(n)] \cup \{-\infty\}$ for all $i \in [n]$, their 
    $(\max,+)$-convolution $A \oplus B$ can be computed in expected time $\widetilde O(n^{1.5})$. 
\end{theorem}

Note that Chi, Duan, Xie and Zhang phrase their result for $(\min,+)$-convolution of increasing sequences 
with entries in $[O(n)]$. We prove in~\Cref{sec:appendix-equivalence} that both statements are equivalent, 
so their result also works for $(\max,+)$-convolution with entries in $[O(n)] \cup \{-\infty\}$. 

\paragraph*{Witnesses}

Let $A[0,\dots,n], B[0,\dots,n]$ be an instance of \BMMaxConv{}. Let $C := A \oplus B$. Given $k \in [2n]$,
we say that $i \in [n]$ is a \emph{witness} for $C[k]$ if $C[k] = A[i] + B[k-i]$. We say that an array $M[0.\dots,2n]$ 
is a \emph{witness array}, if each entry $M[k]$ contains some witness for $C[k]$.

For the general case of \MaxConv{} it is well known (e.g.~\cite{Seidel95,AlonGMN92}) that computing the witness array has the same time complexity as $(\max,+)$-convolution, up to a $\polylog(n)$ overhead. This reduction does not 
immediately apply to \BMMaxConv{} because the sequences might not remain monotone. However, we make it work with some extra care, see \Cref{sec:witnesses} for the proof. 

\begin{restatable}[Witness Finding]{lemma}{witnesses}\label{lem:witness-finding}
    If \BMMaxConv{} can be computed in time $T(n)$, then a witness array $M[0,\dots,2n]$ can be computed in time 
    $\widetilde{O}(T(n))$.
\end{restatable}

\paragraph*{Niceness asumptions on time bounds}\label{sec:niceness}
For all time bounds $T(n)$ in this paper, we make the following niceness assumptions: (1) $T(\widetilde{O}(n)) \le \widetilde{O}(T(n))$,
and (2) $k \cdot T(n) \le O(T(kn))$ for any $k, n \geq 1$. This is satisfied for all natural time bounds of polynomial-time or 
pseudopolynomial-time algorithms, in particular it holds for all functions of the form $T(n) = \Theta(n^{\alpha} \log^\beta n)$ 
for any constants $\alpha \ge 1, \beta \ge 0$.

\section{Exact algorithm for \UnboundedKnapsack{}}\label{sec:exact-unboundedknapsack}
In this section we prove the following Theorem:

\begin{theorem}\label{thm:reduction-unboundedknapsack-to-conv}
    If \BMMaxConv{} on length-$n$ sequences can be solved in time $T(n)$, then \UnboundedKnapsack{} can be solved 
    in time $\widetilde{O}(n + T(p_{\max} + w_{\max}))$, where $p_{\max}$ is the largest profit
    of any item and $w_{\max}$ is the largest weight of any item.
\end{theorem}

Note that~\Cref{thm:alg-unboundedknapsack} follows as an immediate corollary of~\autoref{thm:reduction-unboundedknapsack-to-conv}
by plugging in Chi, Duan, Xie and Zhang's algorithm (\autoref{thm:convolution-bounded}).

\medskip

For the entire section, fix an instance $(\calI, W)$ of the \UnboundedKnapsack{} problem.
Recall that $\calP_i[0,\dots,W]$ is defined as $\calP_i[j] := \max\{p(x) \colon w(x) \leq j, \|x\|_1 \leq 2^i\}$,
and set $\Delta := p_{\max} + w_{\max}$.
Suppose we know that the optimal solution consists of at most $2^k$ items. Then, our goal is to compute the value $\calP_k[W]$.
The natural approach is to use dynamic programming: if we have computed $\calP_{i-1}$, then $\calP_i = \calP_{i-1} \oplus \calP_{i-1}$.
To get our desired running time, we will show that we only need to convolve $O(\Delta)$ entries of $\calP_{i-1}$ and that we can 
enforce that all of these fall in a range of $O(\Delta)$ values. By monotonicity of $\calP_{i-1}$, we end up with a \BMMaxConv{} instance, 
which can be solved in time $O(T(\Delta))$. The resulting total time to compute $\calP_i[W]$ is $O(n + k \cdot T(\Delta))$. 
With additional preprocessing we ensure that $k = O(\log \Delta)$, turning the running time into $\widetilde O(n + T(\Delta))$.

\subsection{Preparations}

We need to show that when computing the optimal answer for some entry $\calP_i[j]$, we can split it 
in such a way that both its total profit and its total weight are roughly halved. Our main tool to show this is the Steinitz 
Lemma~\cite{GrinbergS80,Steinitz1913}. A beautiful proof for it can be found in~\cite{Matousek}.

\begin{lemma}[{\cite[Steinitz Lemma]{GrinbergS80,Steinitz1913}}]\label{lem:steinitz}
    Let $\|.\|$ be a norm in $\R^m$ and let $M$ be an arbitrary collection of $t$ vectors 
    in $\R^m$ such that $\|v\| \leq 1$ for every $v \in M$ and $\sum_{v \in M} v = 0$. Then, 
    it is possible to permute the vectors in $M$ into a sequence $(v_1,\dots,v_t)$ such that 
    $\|v_1 + \dots + v_k \| \leq m$ holds for every $k \in [t]$.
\end{lemma}

We use the Steinitz Lemma to argue that the items in a solution can be split 
in two parts in such a way that both the total profit and the total weight are roughly halved:

\begin{lemma}[Splitting Lemma]\label{lem:splitting-steinitz}
    Let $i \ge 1$ and consider a solution $x \in \N^n$ with $\|x\|_1 \leq 2^i$. 
    Then there is a partition of $x$ into two solutions $x_1, x_2 \in \N^n$ with the following properties:
    \begin{enumerate}
        \item (Splitting of Items) $\|x_1\|_1, \|x_2\|_1 \leq 2^{i-1}$ and $x = x_1 + x_2$,
        \item (Approximate Splitting of Weight) $|w(x_1) - \tfrac{1}{2}w(x)| \leq 2\Delta$ and $|w(x_2) - \tfrac{1}{2}w(x)| \leq 2\Delta$,
        \item (Approximate Splitting of Value) $|p(x_1) - \tfrac{1}{2}p(x)]| \leq 2\Delta$ and $|p(x_2) - \tfrac{1}{2}p(x)]| \leq 2\Delta$.
    \end{enumerate}
\end{lemma}
\begin{proof}
    Let $t := \|x\|_1 \leq 2^i$. First assume that $t$ is even; we will remove this assumption later.
    Write $x = \sum_{j=1}^t x^{(j)}$ where each $x^{(j)}$ corresponds to one copy of some item, i.e. $\|x^{(j)}\|_1 = 1$, 
    and set $v^{(j)} = \left(\begin{smallmatrix} w(x^{(j)}) \\ p(x^{(j)}) \end{smallmatrix} \right)$.  Note that $\|v^{(j)}\|_{\infty} \leq \Delta$. By applying the Steinitz Lemma on the vectors
    $v^{(j)} - \tfrac{1}{t} \left(\begin{smallmatrix} w(x) \\ p(x) \end{smallmatrix} \right)$ (after normalizing by $\Delta$), we can assume that the $v^{(j)}$'s 
    are ordered such that 
    \begin{equation}\label{eqn:splitting-steinitz}
        \bigg \|\sum_{j = 1}^{t/2}v^{(j)} - 
            \frac{1}{2}\begin{pmatrix} w(x) \\ p(x) \end{pmatrix} \bigg \|_{\infty} 
            \leq 2 \Delta.
    \end{equation}
    Fix this ordering, and let $x_1 = x^{(1)} + \ldots + x^{(t/2)}$, corresponding to $v^{(1)},\dots,v^{(t/2)}$, and let $x_2 = x^{(t/2+1)}+\ldots+x^{(t)}$, corresponding to the remaining vectors $v^{(t/2 + 1)},\dots,v^{(t)}$.
    We now check that $x_1, x_2$ satisfy the properties of the lemma:
    
    \begin{itemize}
        \item Property 1 is clearly satisfied by construction.
        \item For property 2, note that~\eqref{eqn:splitting-steinitz} implies $|w(x_1) - \tfrac{1}{2}w(x)| \leq 2\Delta$.
        Since $w(x_2) = w(x) - w(x_1)$, we have that $|w(x_2) - \tfrac{1}{2}w(x)| = |\tfrac{1}{2}w(x) - w(x_1)| \leq 2\Delta$.
        \item Property 3 follows in the same way as property 2.
    \end{itemize}  
    If $t$ is odd, then $t+1 \le 2^i$, so we can add a dummy vector $x^{(t+1)} = 0$ with corresponding $v^{(t+1)} := \left(\begin{smallmatrix} 0 \\ 0 \end{smallmatrix} \right)$ and repeat the same argument with $t := t + 1$.
\end{proof}

When we apply~\Cref{lem:splitting-steinitz} to an optimal solution corresponding to an entry of the array $\calP_i$,
we obtain the following lemma.

\begin{lemma} \label{lem:splitting}
    Let $\beta > 0$. For any index $j \in [\beta \pm 8\Delta] \cap [W]$ there are indices $j_1, j_2 \in [\tfrac{\beta}{2} \pm 8\Delta] \cap [W]$ with the following properties:
    \begin{itemize}
        \item[(i)] $j_1 + j_2 = j$,
        \item[(ii)] $\calP_i[j] = \calP_{i-1}[j_1] + \calP_{i-1}[j_2]$,
        \item[(iii)] $|\calP_{i-1}[j_1] - \tfrac{1}{2}\calP_i[j]| \leq 2\Delta$ and $|\calP_{i-1}[j_2] - \tfrac{1}{2}\calP_i[j]| \leq 2\Delta$.
    \end{itemize}
\end{lemma}
\begin{proof}
    Let $x \in \N^n$ be an optimal solution for $\calP_i[j]$, that is, we have $p(x) = \calP_i[j]$, $w(x) \le j$, and $\|x\|_1 \le 2^i$. 
    We apply~\Cref{lem:splitting-steinitz} to $x$ and obtain $x_1, x_2 \in \N^n$ such that $x_1 + x_2 = x$.
    We do a case distinction based on $w(x_1), w(x_2)$:
    \begin{itemize}
        \item $w(x_1), w(x_2) \in [\tfrac{j}{2} \pm 4\Delta]$: Let $j_1 := w(x_1)$ and $j_2 := j - j_1$; note that $j_1,j_2 \in [\tfrac j 2 \pm 4 \Delta] \subseteq [\tfrac \beta 2 \pm 8 \Delta]$. 
        We argue that $p(x_1) = \calP_{i-1}[j_1]$ and $p(x_2) = \calP_{i-1}[j_2]$. Indeed, since $w(x_1) = j_1$ the solution $x_1$ is feasible for weight $j_1$, so $p(x_1) \le \calP_{i-1}[j_1]$. Similarly, since $w(x_2) = w(x)-w(x_1) \le j - w(x_1) = j_2$ the solution $x_2$ is feasible for weight $j_2$, so $p(x_2) \le \calP_{i-1}[j_2]$. Moreover, by optimality of $x$ we have $p(x_1) + p(x_2) = p(x) = \calP_i[j] \ge \calP_{i-1}[j_1] + \calP_{i-1}[j_2]$, so we obtain $p(x_1) = \calP_{i-1}[j_1]$ and $p(x_2) = \calP_{i-1}[j_2]$. Using these equations together with $p(x) = \calP_i[j]$, property (ii) follows from $p(x) = p(x_1)+p(x_2)$, property (iii) follows from Property 3 of~\Cref{lem:splitting-steinitz}, and property (i) holds by definition of $j_2$. 
        
        \item $w(x_1) < \tfrac{j}{2} - 4\Delta$: Property 2 of~\Cref{lem:splitting-steinitz} implies that 
        $|w(x_1) - w(x_2)| \leq 4\Delta$, and thus $w(x_2) \leq \tfrac{j}{2}$. 
        Therefore, $x_1$ and $x_2$ are feasible for weights $j_1 := \lfloor \tfrac j 2 \rfloor$ and $j_2 := \lceil \tfrac j 2 \rceil$, respectively. Note that $j_1,j_2 \in [\tfrac \beta 2 \pm (4 \Delta + 1)] \subseteq [\tfrac \beta 2 \pm 8 \Delta]$.
        Property (i) is obvious, and properties (ii) and (iii) now follow as in the first case.

        \item $w(x_1) > \tfrac{j}{2} + 4\Delta$: Similarly as the previous case, property 2 of~\Cref{lem:splitting-steinitz}
        implies that $w(x_2) \geq \tfrac{j}{2}$. Therefore, we have
        $w(x) = w(x_1) + w(x_2) > j + 4\Delta$, which contradicts the assumption $w(x) \leq j$.
        
        \item $w(x_2) < \tfrac{j}{2} - 4\Delta$ or $w(x_2) > \tfrac{j}{2} + 4\Delta$: Symmetric to the previous two cases.\qedhere
    \end{itemize}
\end{proof}

\subsection{The algorithm}

We are now ready to present our algorithm. The idea is to use the Splitting~\Cref{lem:splitting-steinitz} to 
convolve smaller sequences which are bounded and monotone.

Let $\calI = \{(w_i, p_i)\}_{i=1}^n$ with capacity constraint $W$ be an instance of \UnboundedKnapsack{}.
Since any item has $w_i \geq 1$, we know that any solution $x \in \N^n$ consists of at most $W$ items.
Thus, to compute the value of the optimal solution it suffices to compute $\calP_k[W]$ where $k := \lceil \log W \rceil$.

Our approach is as follows. We do binary search for $\OPT$ in the range $[p_{\max} \cdot W]$. Suppose we have the current guess $\alpha$.
Instead of computing the arrays $\calP_i$, we compute \emph{clipped versions}, i.e., $C_i$ which has the property that 
$C_i[j] \geq \alpha$ if and only if $\calP_i[j] \geq \alpha$. 

We compute $C_i$ as follows: At every step, we only compute the values for $O(\Delta)$ weights
$C_i[W\cdot 2^{i-k} \pm 8\Delta]$. For the base case $i = 0$, we simply set $C_0[0,\dots,8\Delta] := \calP_0[0,\dots,8\Delta]$.
Note that this can be done in time $O(n+\Delta)$ by doing one pass over the item set, since $\calP_0$ only considers 
solutions with at most one item. Moreover, observe that $C_0[0,\dots,8\Delta]$ is monotone non-decreasing
by definition of $\calP_0$.

For the general case $i > 0$ we first compute an array $A_i[W\cdot 2^{i-k} \pm 8\Delta]$ by taking the 
$(\max, +)$-convolution of $C_{i-1}[W \cdot 2^{i-1-k} \pm 8\Delta]$ with itself.
To obtain $C_i[W\cdot 2^{i-k} \pm 8\Delta]$, we clip the values in $A_i$ which are too large, and set to 
$-\infty$ the values which are to small. This ensures that all values in $C_i$ lie
within a range of $O(\Delta)$, except for values that are $-\infty$. \Cref{alg:steinitz-knapsack} contains the pseudocode.

\begin{algorithm}
    \caption{Given an instance $(\calI, W)$ of \UnboundedKnapsack{} and a guess $\alpha \in [p_{\max} \cdot W]$, the algorithm computes a
    value $C_k[W]$ satisfying the guarantee in~\Cref{lem:guarantee-alg-steinitz}.}\label{alg:steinitz-knapsack}
    \begin{algorithmic}[1]
        \State $k := \lceil \log W \rceil$
        \State Initialize $C_0[0,\dots,8\Delta] := \calP_0[0,\dots,8\Delta]$ by iterating over the item set $\calI$ once
        \For{$i = 1,\dots,k$}
            \State $A_i[W\cdot 2^{i-k}\pm8\Delta] := 
                C_{i-1}[W\cdot 2^{i-1-k}\pm 8\Delta] \oplus 
                C_{i-1}[W\cdot 2^{i-1-k}\pm 8\Delta]$
            \State $C_i[j] := 
                \begin{cases}
                    \lceil \alpha \cdot 2^{i-k} \rceil + 24\Delta & \text{if } A_i[j] > \alpha \cdot 2^{i-k} + 24\Delta \\
                    -\infty & \text{if } A_i[j] < \alpha \cdot 2^{i-k} - 40\Delta \\
                    A_i[j] & \text{otherwise}
                \end{cases}$
        \EndFor
        \Return $C_k[W]$
    \end{algorithmic}
\end{algorithm}

Due to the clipping, at every step we compute a $(\max, +)$-convolution of sequences of length 
$O(\Delta)$ and values in $[O(\Delta)] \cup \{-\infty\}$ (after shifting the indices and values appropriately). Furthermore, note 
that all convolutions involve monotone non-decreasing sequences. Indeed, as noted above the starting sequence $C_0$ is 
monotone non-decreasing. Convolving it with itself produces a monotone non-decreasing sequence again, and the clipping 
in line 5 of~\Cref{alg:steinitz-knapsack} preserves monotonicity. The same argument applies for 
further iterations. Thus, the running time of~\Cref{alg:steinitz-knapsack} is $O(n + T(\Delta)\log W)$,
where $T(\Delta)$ is the running time to compute \BMMaxConv{} on sequences of length~$\Delta$.

Regarding correctness, we claim the following:

\begin{claim}
    For every $i \in [k]$ and every index $j \in [W\cdot 2^{i-k} \pm 8\Delta] \cap [W]$ the following holds:
    \begin{itemize}
        \item If $\calP_i[j] \in [\alpha \cdot 2^{i-k} - 40\Delta, \alpha\cdot 2^{i-k} + 24\Delta]$, then $C_i[j] = \calP_i[j]$.
        \item If $\calP_i[j] > \alpha \cdot 2^{i-k} + 24\Delta$, then $C_i[j] = \lceil \alpha \cdot 2^{i-k} \rceil + 24\Delta$.
        \item If $\calP_i[j] < \alpha \cdot 2^{i-k} - 40\Delta$, then $C_i[j] = -\infty$.
    \end{itemize}
\end{claim}
Intuitively, the claim says that entries ``close'' to the (scaled) guess $\alpha \cdot 2^{i-k}$ get computed exactly,
while entries below and above get clipped appropriately.

\begin{proof}
    We prove the claim by induction on $i$. In the base case $i = 0$, note that since $\alpha \in [p_{\max}\cdot W]$ and $k = \lceil \log W \rceil$ we have $\alpha \cdot 2^{-k} \le p_{\max} \le \Delta$. Thus, $[\alpha \cdot 2^{-k} - 40 \Delta, \alpha \cdot 2^{-k} + 24 \Delta]$ contains the whole interval $[0,\Delta]$ of possible values of $\calP_0[j]  = C_0[j]$ (for any $0 \le j \le 8\Delta$).

    Now we show that the claim holds for any $1 \le i \leq k$ assuming it holds for $i - 1$. Fix any $j \in [W\cdot2^{i-k}\pm 8\Delta]$.
    Note that the thresholding in line 4 of~\Cref{alg:steinitz-knapsack} does not increasy any of the entries in $A_i$,
    so $C_i[j] \leq A_i[j]$. 
    Moreover, since inductively $C_{i-1}[j'] \le \calP_{i-1}[j']$ for all $j'$, by definition of $A_i$ we have $A_i[j] \le \calP_i[j]$. 
    Hence, we obtain $C_i[j] \leq \calP_i[j]$.
    We use this observation to obtain the claim, by showing an appropriate lower bound for $C_i[j]$ in the following. 
    
    Pick indices $j_1, j_2$ as guaranteed by~\Cref{lem:splitting}.
    Property (i) of~\Cref{lem:splitting} guarantees that the computation of $A_i[j]$ in line 3 of~\Cref{alg:steinitz-knapsack}
    looks at the entries $j_1, j_2$ in $C_{i-1}$. Hence, $A_i[j] \geq C_{i-1}[j_1] + C_{i-1}[j_2]$.

    We proceed by a case 
    distinction on the values of the entries $\calP_{i-1}[j_1]$ and $\calP_{i-1}[j_2]$:

    \paragraph*{Case 1: $\calP_{i-1}[j_1], \calP_{i-1}[j_2] \in [\alpha \cdot 2^{i-1-k}-40\Delta, \alpha\cdot 2^{i-1-k} + 24\Delta]$}
    By the induction hypothesis, both values are computed exactly, that is, $C_{i-1}[j_1] = \calP_{i-1}[j_1]$ and 
    $C_{i-1}[j_2] = \calP_{i-1}[j_2]$. Thus, $A_i[j] \geq C_{i-1}[j_1] + C_{i-1}[j_2] = \calP_{i-1}[j_1] + \calP_{i-1}[j_2] = \calP_i[j]$, using property (ii) of~\Cref{lem:splitting}. Since we observed above that
    $A_i[j] \leq \calP_i[j]$, we obtain $A_i[j] = \calP_i[j]$. The thresholding in line 4 of~\Cref{alg:steinitz-knapsack}
    now yields the claim for this case.

    \paragraph*{Case 2: $\calP_{i-1}[j_1] > \alpha\cdot 2^{i-1-k} + 24\Delta$}
    Property (iii) of~\Cref{lem:splitting} implies that $|\calP_{i-1}[j_1] - \calP_{i-1}[j_2]| \leq 4\Delta$, and hence
    $\calP_{i-1}[j_2] \geq \alpha \cdot 2^{i-1-k} + 20\Delta$. Thus, property (ii) of~\Cref{lem:splitting} implies $\calP_i[j] = \calP_{i-1}[j_1] + \calP_{i-1}[j_2] > \alpha \cdot 2^{i-k} + 24\Delta$.
    Therefore, we want to show that $C_i[j] = \lceil \alpha \cdot 2^{i-k} \rceil + 24\Delta$.

    By the induction hypothesis, we have $C_{i-1}[j_1] = \lceil \alpha \cdot 2^{i-1-k} \rceil + 24\Delta$ and 
    $C_{i-1}[j_2] \geq \alpha \cdot 2^{i-1-k} + 20\Delta$. Hence, 
    $A_i[j] \geq C_{i-1}[j_1] + C_{i-1}[j_2] > \alpha \cdot 2^{i-k} + 40\Delta$. Due to the thresholding, we conclude that
    $C_i[j] = \lceil \alpha \cdot 2^{i-k} \rceil + 24\Delta$, as desired.

    \paragraph*{Case 3: $\calP_{i-1}[j_1] < \alpha \cdot 2^{i-k} - 40\Delta$} Similarly as in case 2, property (iii) of~\Cref{lem:splitting} implies 
    that $\calP_{i-1}[j_2] \leq \alpha \cdot 2^{i-1-k} - 36\Delta$. Thus, $\calP_i[j] = \calP_{i-1}[j_1] + \calP_{i-1}[j_2] < \alpha \cdot 2^{i-k} - 76\Delta$. 
    Since $C_i[j] \leq \calP_i[j]$, but $C_i[j]$ takes values in $\{-\infty\} \cup [\alpha \cdot 2^{i-k} \pm 40\Delta]$ it follows that 
    $C_i[j] = -\infty$. \qedhere
\end{proof}

Given the claim, it is easy to see that $C_k[W] \geq \alpha$ if and only if $\calP_k[W] \geq \alpha$. Along with 
the running time analysis argued earlier, we obtain the following lemma.
\begin{lemma}\label{lem:guarantee-alg-steinitz}
    \Cref{alg:steinitz-knapsack} runs in time $O(n + T(\Delta) \log W)$, where $T(\Delta)$ is the time complexity of 
    \BMMaxConv{} on sequences of length $\Delta$, and computes a value
    $C_k[W]$ which satisfies $C_k[W] \geq \alpha$ if and only if $\OPT = \calP_k[W] \geq \alpha$.
\end{lemma}

Given~\Cref{lem:guarantee-alg-steinitz}, we can do binary search to find the optimal value. This gives an algorithm for
\UnboundedKnapsack{} in time $O((n+T(\Delta)) \log W \log \OPT)$. To shave the $\polylog(W, \OPT)$ factors and 
obtain the running time $\widetilde{O}(n + T(\Delta))$ claimed 
in~\Cref{thm:reduction-unboundedknapsack-to-conv}, we make use of the following lemma. It allows us to reduce the 
capacity of the instance by repeatedly adding copies of the item with maximum profit-to-weight ratio.
Similar results have been shown for general ILPs~\cite{EisenbrandW20},
for \UnboundedKnapsack{}~\cite{BateniHSS18} and for the Coin Change problem~\cite{ChanH22}. For completeness, 
we include the proof by Chan and He~\cite[Lemma 4.1]{ChanH22}.\footnote{Both Chan and He~\cite{ChanH22} 
and Bateni et al.~\cite{BateniHSS18} show that the same conclusion of the lemma holds if $W > w_{i^*}^2$, with a slightly more 
involved argument. For our purposes, this simple variant is enough.}
\begin{lemma}\label{lem:greedy}
    Let $(p_{i^*}, w_{i^*}) := \argmax_{(p,w) \in \calI} \tfrac{p}{w}$. If $W \geq 2w_{\max}^3$, then there exists an optimal 
    solution containing $(p_{i^*}, w_{i^*})$.
\end{lemma}
\begin{proof}
    Consider an optimal solution $x$ that does not contain item $(p_{i^*}, w_{i^*})$. If there is an item $(p_j, w_j)$ that appears at least $w_{i^*}$ times in $x$, then we can replace 
    $w_{i^*}$ of the copies of item $(p_j, w_j)$ by $w_j$ copies of item $(p_{i^*}, w_{i^*})$. By definition of $(p_{i^*}, w_{i^*})$, 
    this does not decrease the total profit of the solution, so by optimality of $x$ the new solution $x'$ is also optimal. Therefore, some optimal solution contains $(p_{i^*}, w_{i^*})$.
    
    It remains to consider the case that $x$ contains less than $w_{i^*}$ copies of every item, so its total weight is at most $n \cdot w_{i^*} \cdot w_{\max}$.
    Note  that $n \leq w_{\max}$, because without loss of generality there is at most one item per distinct weight (otherwise we can keep only the item with 
    the largest profit for each weight). Thus, the total weight of $x$ is at most $w_{\max}^3$. It follows that $W < w_{\max}^3 + w_{i^*} \leq 2w_{\max}^3$,
    since otherwise we could add at least one copy of $(p_{i^*}, w_{i^*})$ to $x$, contradicting its optimality.
\end{proof}

We now put all pieces together: As a preprocessing step we repeatedly add the item $(p_{i^*}, w_{i^*})$ and decrease $W$ by $w_{i^*}$,
as long as $W > 2w_{\max}^3$. After this preprocessing, we have $W = O(w_{\max}^3) = O(\Delta^3)$, and thus
$\OPT \leq W \cdot p_{\max} = O(\Delta^4)$. The we do binary search for $\OPT$, using~\Cref{alg:steinitz-knapsack} as a decision procedure.
By~\Cref{lem:guarantee-alg-steinitz}, the overall running time is $O((n + T(\Delta)) \log^2 \Delta) = \widetilde{O}(n + T(\Delta))$. 
This completes the proof of~\Cref{thm:reduction-unboundedknapsack-to-conv}.

\subsection{Solution Reconstruction}

The algorithm we described gives us the value $\OPT$ of the optimal solution. In this section we will describe how to 
use witness arrays (\Cref{lem:witness-finding}) to
reconstruct a feasible solution $x \in \N^n$ such that $p(x) = \OPT$ with only a polylogarithmic overhead in the overall running time.

\begin{lemma}
    A optimal solution $x$ can be reconstructed in time $\widetilde{O}(n + T(p_{\max} + w_{\max}))$.
\end{lemma}
\begin{proof}[Proof Sketch]
    Let $k = \lceil \log W \rceil$ be as in~\Cref{alg:steinitz-knapsack}.
    After determining the value of $\OPT$, run~\Cref{alg:steinitz-knapsack} again with the 
    guess $\alpha = \OPT$. For every \BMMaxConv{} in line 4 compute the witness array $M_i$ corresponding to
    $A_i$ via~\Cref{lem:witness-finding}. This takes time $\widetilde{O}(n + T(p_{\max} + w_{\max}))$. 
    Now, the idea is to start from $C_k[W]$ and traverse the \emph{computation tree} 
    of~\Cref{alg:steinitz-knapsack} backwards. That is, we look at the pair of entries $C_{k-1}[M_k[W]], C_{k-1}[W - M_k[W]]$
    which define the value of $C_k[W]$ and recursively obtain the pair of entries in $C_{k-2}$ determining the value 
    of $C_{k-1}[M_k[W]]$, etc. By proceeding in this way, we eventually hit the leaves, i.e.,
    the entries of $C_0[0,\dots,8\Delta] = \calP_0[0,\dots,8\Delta]$, which correspond to the items 
    in an optimal solution. A naive implementation of this idea takes time $O(\sum_{i \leq k}2^i) = O(2^k) = O(W)$, which is too slow.

    Now we describe an efficient implementation of the same idea. For each $i \in [k]$ construct an array 
    $Z_i[W \cdot 2^{i-k} \pm 8\Delta]$ initialized to zeros. Set $Z_k[W] := 1$. We will maintain the invariant that $Z_i[j]$ 
    stores the number of times we arrive at $C_i[j]$ by traversing the computation tree starting at $C_k[W]$.
    This clearly holds for $Z_k[W] = 1$ by definition. Now we describe how to fill the entries for the levels below.
    Iterate over $i = k,k-1,\dots,1$. For each entry 
    $j \in [W \cdot 2^{i-k} \pm 8\Delta] \cap [W]$ add $Z_i[j]$ to its witness entries in the level below, i.e., 
    increase $Z_{i-1}[M_i[j]]$ by $Z_i[j]$ and $Z_{i-1}[j-M_i[j]]$ by $Z_i[j]$.
    The invariant is maintained by definition of the witnesses, and because~\Cref{alg:steinitz-knapsack} guarantees that 
    $M_i[j], j-M_i[j] \in [2 \cdot 2^{i-1-k} \pm 8\Delta] \cap [W]$. Note that this procedure takes time 
    $O(k \Delta) = \widetilde{O}(\Delta)$.

    Finally, note that for the base case we have that each $Z_0[j]$ for $j \in [8\Delta]$ counts the number of times that we hit the 
    entry $C_0[j] = \calP_0[j]$ in the computation tree starting from $C_k[W]$. Recall that by definition, 
    $\calP_0[j]$ is the maximum profit of an item in $\calI$ with weight at most $j$. Hence, every entry 
    $\calP_0[j]$ corresponds to a unique item in $\calI$. Therefore, we can read off from $Z_0[0,\dots,8\Delta]$ the multiplicity 
    of each item included in an optimal solution. The overall time of the procedure is $\widetilde{O}(n + T(p_{\max} + w_{\max}))$,
    as claimed.
\end{proof}

\section{Exact Algorithm for \Knapsack{}}\label{sec:exactzerooneknapsack}

Cygan et al.~\cite{CyganMWW19} showed the following reduction from \Knapsack{} to \BMMaxConv{}:

\begin{theorem}[{\cite[Theorem 13]{CyganMWW19}}]\label{thm:reduction-knapsack-cygan}
	If \MaxConv{} on length-$n$ sequences can be solved in time $T(n)$, then 
	\Knapsack{} can be solved in time $O(T(W \log W) \log^3(n/\delta) \log n)$ with probability 
	at least $1 - \delta$.
\end{theorem}

Their reduction is a generalization of Bringmann's algorithm for \SubsetSum{}~\cite{Bringmann17}, replacing Boolean 
convolutions by $(\max,+)$-convolutions. We observe that essentially the same reduction yields 
sequences of length $O(W)$ which are monotone and have entries bounded by $\OPT$. In particular, these are \BMMaxConv{}
instances. This yields the following:

\begin{theorem}[$\Knapsack{} \rightarrow \BMMaxConv{}$]\label{lem:reduction-0-1knapsack-to-maxconv}
    If \BMMaxConv{} on length-$n$ sequences can be solved in time $T(n)$, then 
    \Knapsack{} can be solved in time $\widetilde{O}(n + T(W + \OPT))$ with high probability.
\end{theorem}
\begin{proof}
    The proof is virtually the same as~\cite[Theorem 13]{CyganMWW19}, so we omit some details. In particular,
    we emphasize how the constructed instances can be seen to be monotone and bounded, but we omit some details of the correctness argument.
    The idea of the algorithm is the following: split the item set $\calI$ into groups $G_{(a,b)} \subseteq \calI$
    such that all items $(p, w) \in \calI$ with $p \in [2^{a-1}, 2^a)$ and $w \in [2^{b-1}, 2^b)$ are in group $G_{(a,b)}$. That is,
    all items within each group have weights and profits within a factor of 2 of each other, and thus there are $O(\log W \log \OPT)$
    many groups. We will describe how to compute $\calP_{G_{(a,b)}}[0,\dots,W]$ for each $G_{(a,b)}$. Having that, we simply combine 
    all the profit arrays into $\calP_{\calI}[0,\dots,W]$ using $(\max,+)$-convolutions. Since we have $O(\log W \log \OPT)$ groups, and each profit array is a monotone non-decreasing sequence of length $W$ with entries bounded by $\OPT$, the combination step takes 
    time $\widetilde{O}(T(W + \OPT))$.

    Fix some group $G_{(a,b)}$. Since every $(p,w) \in G_{(a,b)}$ has $w \in [2^{a-1}, 2^a)$ and $p \in [2^{b-1}, 2^b)$, any 
    feasible solution from $G_{(a,b)}$ consists of at most $z := \lceil \min\{W/2^{a-1}, \OPT/2^{b-1}\} \rceil$ items. Thus, by splitting the items in
    $G_{(a,b)}$ randomly into $z$ subgroups $G_{(a,b),1},\dots,G_{(a,b),z}$, any fixed feasible solution has at most $O(\log z)$ items in 
    each subgroup $G_{(a,b), k}$ with high probability. To see this, fix a solution $x$ and note that, 
    \begin{align*}
        \Pr[\text{at least } r \text{ items from }x \text{ fall in } G_{(a,b), k}] 
            &\leq {z \choose r}\left(\frac{1}{z}\right)^r
            \leq \left(\frac{e \cdot z}{r}\right)^r \left(\frac{1}{z}\right)^r = \left(\frac{e}{r}\right)^r,
    \end{align*}
    where the first inequality follows due to a union bound over all subsets of items of size $r$.
    By setting $r = O(\log z)$, we can bound this probability by $z^{-c}$ for any constant $c$. So by a union bound,
    none of the $z$ groups $G_{(a,b), 1},\dots,G_{(a,b), z}$ has more than $\kappa := O(\log z)$ elements from the fixed solution $x$
    with probability at least $1 - 1/\poly(z)$.
    
    Therefore, to obtain the value of any fixed solution it suffices to compute the 
    optimal solution consisting of at most $\kappa$ items from $G_{(a,b), k}$ for every target weight $\leq O(2^a \kappa)$, and then merge the results. 
    More precisely, for every $1 \leq i \leq z$ we compute the array $\calP_{G_{(a,b), i}, \log(\kappa)}[0,\dots,O(2^a \kappa)]$.
    Recall that this is defined as 
    \[
        \calP_{G_{(a,b),i},\log(\kappa)}[j] := \max\{p(x) \colon x \text{ is a solution from } G_{(a,b),i} \text{ with } w(x) \leq j,  \|x\|_1 \leq \kappa \}
    \]
    for each $j \in [O(2^a \kappa)]$. For ease of notation, we denote the array by $\calP_{G_i, \kappa} := \calP_{G_{(a,b),i},\log \kappa}$.
    
    Then, we merge the $\calP_{G_i,\kappa}$'s using $(\max,+)$-convolutions. Now we describe these two steps in more detail:

    \paragraph*{Computing $\calP_{G_i,\kappa}[0,\dots, O(2^a \kappa)]$} Since we only care about solutions with at most $\kappa$ items, we 
    use randomization again\footnote{This step is called ``Color Coding'' in~\cite{Bringmann17,CyganMWW19}.}: split the items in $G_{(a,b), i}$ 
    into $\kappa^2$ buckets $A_1,\dots,A_{\kappa^2}$. By the birthday paradox, with constant 
    probability it holds that any fixed solution is shattered among the buckets, i.e., each bucket contains at most 1 item of the solution.
    Thus, for each bucket $A_{k}$ we construct the array $\calP_{A_k, 0}[0,\dots,2^a]$. Recall that this is defined as
    \[
        \calP_{A_k, 0}[j] := \max\{p(x) \colon x \text{ is a solution from } A_k \text{ with } w(x) \leq j, \|x\|_1 \leq 1\}
    \]
    for each entry $j \in [2^a]$.
    To combine the results, we compute $\calP_{A_1,0} \oplus \calP_{A_2,0} \oplus \dots \oplus \calP_{A_{\kappa^2}, 0}$. 
    By definition, every $\calP_{A_i,0}$ is a monotone non-decreasing sequence of length $2^a$ with entries bounded by $2^b$. Thus, 
    the merging step takes time $O(T((2^a + 2^b) \cdot \kappa^2) \cdot \kappa^2)$.

    Each entry of the resulting array has the correct value $\calP_{i,\kappa}[j]$ with constant probability, since a corresponding optimal
    solution is shattered with constant probability. By repeating this process $O(\log z)$ times and keeping the entrywise maximum among 
    all repetitions, we boost the success probability to $1 - 1/\poly(z)$. Thus, by a union bound over the $z$ subgroups 
    $G_{(a,b),1},\dots,G_{(a,b), z}$, we get that any $z$ fixed entries $\calP_{G_1,\kappa}[j_1],\dots,\calP_{G_z,\kappa}[j_z]$
    corresponding to a solution which is partitioned among the $z$ subgroups get computed correctly with probability at least $1 - 1/\poly(z)$.
    This adds an extra $O(\log z) = O(\kappa)$ factor to the running time.

    \paragraph*{Merging $\calP_{G_1,\kappa} \oplus \dots \oplus \calP_{G_{z},\kappa}$} This computation is done in a binary tree-like fashion. That is, 
    in the first level we compute $(\calP_{G_1,\kappa} \oplus \calP_{G_2,\kappa}), (\calP_{G_3, \kappa} \oplus \calP_{G_4,\kappa}),\dots,(\calP_{G_{z-1},\kappa} \oplus \calP_{G_z, \kappa})$. 
    In the second level we merge the results from the first level in a similar way. We proceed in the same way with further levels. 
    Since we merge $z$ sequences,
    we have $\lceil \log z \rceil$ levels of computation. In the $j$-th level, we compute the $(\max, +)$-convolution of 
    $z / 2^j$ many monotone non-decreasing sequences of length $O(2^j \cdot 2^a \cdot \kappa)$ with entries bounded by 
    $O(2^j \cdot 2^{b} \cdot \kappa)$. 
    Therefore, overall the merging takes time 
    \[
       O\bigg(\sum_{j=1}^{\lceil \log z \rceil} \frac{z}{2^j} \cdot T((2^a + 2^b)\cdot 2^j \cdot \kappa) \bigg)
       \leq \widetilde{O}(T((2^a + 2^b) \cdot z)),
    \]
    where we used both of our niceness assumptions $k \cdot T(n) \leq O(T(k \cdot n))$ for any $k > 1$ and $T(\widetilde{O}(n)) \leq \widetilde{O}(T(n))$.
    Since $z = \lceil \min\{W/2^{a-1}, \OPT/2^{b-1}\} \rceil$, we have $\widetilde{O}(T((2^a + 2^b) \cdot z) = \widetilde{O}(T(W + \OPT))$.
    
    \paragraph*{Wrapping up} To recap, the algorithm does the following steps:
    \begin{enumerate}
        \item Split the items into $O(\log W \log \OPT)$ groups $G_{(a,b)}$. This takes time $O(n)$.
        \item Randomly split each group $G_{(a,b)}$ into $z := \lceil \min\{W/2^{a-1}, \OPT/2^{b-1}\} \rceil$ 
        subgroups $G_{(a,b), i}$ for $i \in [z]$.
        \item For each $G_{(a,b), i}$ compute the array $\calP_{G_i, \kappa}[0,\dots,O(2^a \kappa)]$ in time $O(T((2^a + 2^b) \kappa^2) \cdot \kappa^3)$.
        Since $\kappa = O(\log z)$, the total time over all $i \in [z]$ is 
        \[
            O(z \cdot T((2^a + 2^b) \kappa^2) \cdot  \kappa^3) \leq 
                O(T((2^a + 2^b) \cdot z \cdot \kappa^2) \kappa^3) \leq 
                \widetilde{O}(T((2^a + 2^b) \cdot z)) \leq \widetilde{O}(T(W + \OPT)).
        \]
        Note that here we use the niceness assumptions on $T(n)$. In particular, first we used that 
        $k \cdot T(n) \leq O(T(k \cdot n))$ for any $k > 1$, and then that $T(\widetilde{O}(n)) \leq \widetilde{O}(T(n))$.
        \item Merge the arrays $\calP_{G_1,\kappa} \oplus \dots \oplus \calP_{G_z, \kappa}$ in time $\widetilde{O}(T(W + \OPT))$ to 
        obtain $\calP_{G_{(a,b)}}[0,\dots,W]$.
        \item Merge the arrays $\calP_{G_{(a,b)}}$ using $O(\log W \log \OPT)$ convolutions in total time $\widetilde{O}(T(W + \OPT))$.
    \end{enumerate}
    Thus, the overall time of the algorithm is $\widetilde{O}(n + T(W + \OPT))$. Note that as mentioned earlier in the proof,
    the algorithm succeeds in computing any \emph{fixed} entry $\calP_{\calI}[j]$ with probability at least $1 - 1/\poly(z)$. In particular,
    this is sufficient to compute the optimal solution $\calP_{\calI}[W]$ with good probability.

    As described, the algorithm only returns the value of the optimal solution. We can easily reconstruct an optimal solution $x \in \N^n$
    for which $p(x) = \OPT$ as we sketch now. Note that at the end of the algorithm, we will have a sequence whose entries correspond to 
    $\calP[0,\dots,W]$. This sequence was obtained as the \BMMaxConv{} of two distinct other sequences, call them $A[0,\dots,W], B[0,\dots,W]$. 
    For the output entry $\calP[W]$, we can find its witness $i \in [W]$, i.e. the index $i$ such that $\calP[W] = A[i] + B[W - i]$. Note that we can find 
    $i$ in time $O(W)$ by simply trying all possibilities. Then, we continue recursively finding the witnesses for $i$ and for $W - i$.
    Eventually, we will reach the entries of the arrays $\calP_{A_k,0}$ which correspond to single items from $\calI$, and these 
    form the solution $x$. The crucial observation is that because the algorithm never convolves a sequence with itself (unlike our 
    algorithm for \UnboundedKnapsack{}), this recursive process finds at most one witness per convolution. Hence, the total time spent 
    is proportional to the total length of the convolved sequences $\widetilde{O}(W)$.
\end{proof}


\section{Approximation Schemes for \UnboundedKnapsack{}}\label{sec:approx}

\subsection{Preparations}

\paragraph*{Greedy 2-approximation}

The fractional solution for an instance $(\calI, W)$ of \UnboundedKnapsack{} has a simple structure: 
pack the entire capacity $W$ greedily with the \emph{most efficient item} $i^*$. That is, choose the item
$(p_{i^*}, w_{i^*}) \in \calI$ which maximizes the ratio $p_{i^*} / w_{i^*}$ and add it $W / w_{i^*}$ many times. 
Since $\lfloor W/w_{i^*} \rfloor$ copies forms a feasible integral solution, it holds that
$\lfloor W / w_{i^*} \rfloor \cdot p_{i^*} \leq \OPT \leq (W/w_{i^*}) \cdot p_{i^*}$.  
By $W / w_{i^*} \geq 1$, we have $\lfloor W / w_{i^*} \rfloor \geq 1/2 \cdot W / w_{i^*}$.
Thus, the greedy solution is a 2-approximation to $\OPT$. Note that we can find this solution in time $O(n)$.

\paragraph*{Connection to \MaxConv{}}

Recall that for an integer $i \geq 0$, we defined $\calP_i[0,\dots,W]$ as $\calP_i[j] := \max\{p(x) \colon w(x) \leq j, \|x\|_1 \leq 2^i\}$.
The following lemma shows that if we know $\calP_{i-1}$, we can compute $\calP_i$ by applying a $(\max,+)$-convolution of $\calP_{i-1}$ with 
itself.

\begin{lemma}[Halving Lemma]\label{lem:splitting-lemma}
    For any $i \geq 1$ and $w \geq 0$, it holds that 
    \[
        \calP_i[0,\dots,w] = \calP_{i-1}[0,\dots,w] \oplus \calP_{i-1}[0,\dots,w].
    \]
\end{lemma}
\begin{proof}
    We denote the right hand side by $C := \calP_{i-1}[0,\dots,w] \oplus \calP_{i-1}[0,\dots,w]$.
    Fix some $j \in [w]$. We will show that $\calP_i[j] = C[j]$: 
    
    $\calP_i[j] \geq C[j]$: This holds since $C[j]$ corresponds to the profit of some solution with at most $2^i$ items, but 
    by definition $\calP_i[j]$ is the maximum profit among all such solutions.

    $\calP_i[j] \leq C[j]$: Let $x \in \N^n$ be the solution corresponding to $\calP_i[j]$, i.e. such that $p(x) = \calP_i[j]$, $w(x) \leq j$
    and $\|x\|_1 \leq 2^i$. Split $x$ in two solutions $x = x_1 + x_2$ of roughly the same size $\|x_1\| \approx \|x_2\|_1 \approx  \|x\|_1/2$ 
    (if $\|x\|_1$ is odd, put an extra item to $x_1$ or $x_2$). In this way, it holds that both $x_1$ and $x_2$ are solutions of at most $2^{i-1}$
    items and weight at most $j$ and therefore by the optimality of $\calP_{i-1}$ it holds that 
    $p(x_1) \leq \calP_{i-1}[w(x_1)]$ and $p(x_2) \leq \calP_{i-1}[w(x_2)]$. Hence, by definition of $C$ we conclude that 
    $C[j] \geq \calP_{i-1}[j - w(x_2)] + \calP_{i-1}[w(x_2)] \geq p(x_1) + p(x_2) = \calP_i[j]$.
\end{proof}

Since any item has weight at least 1, the optimal solution consists of 
at most $W$ items. Thus, we can use~\Cref{lem:splitting-lemma} to compute $\OPT = \calP_{\lceil \log W \rceil}[W]$ 
by applying $O(\log W)$ $(\max,+)$-convolutions (note that the base case $\calP_0[0,\dots,W]$ corresponds to solutions 
of at most one item, and thus can be computed easily). Our approximation schemes will implement this ``repeated squaring'' algorithm 
by appropriately approximating the $(\max,+)$-convolutions.

\subsection{A simplified FPTAS}\label{sec:fptas}

In this section we prove the following reduction from approximating \UnboundedKnapsack{} to \MaxConv{}.

\begin{theorem}\label{thm:reduction-fptas-to-minconv}
    If \MaxConv{} can be solved in time $T(n)$, then \UnboundedKnapsack{} has 
    an FPTAS in time $\widetilde{O}(n + T(1/\varepsilon))$.
\end{theorem}

This proves~\Cref{thm:fptas} by plugging in the bound $n^2/2^{\Omega(\sqrt{\log n})}$ for \MaxConv{}~\cite{BremnerCDEHILPT14,Williams18}.

\subsubsection{Preprocessing}

To give the FPTAS we start with some preprocessing to remove items with small profit and items with small weight.

\paragraph*{Step 1: Remove cheap items} 
We first remove items with small profit.
Let $P_0$ be the value of the greedy 2-approximation and set $T := 2\varepsilon P_0$. 
Split the items into \emph{expensive} $\calI_E := \{(p,w) \in \calI \colon p > T \}$, and 
\emph{cheap} $\calI_C := \calI \setminus \calI_E$. Let $i^* := \argmax_{i \in \calI_C} \frac{p_i}{w_i}$ be the item 
corresponding to the greedy 2-approximation for the cheap items and set $r := \lceil T / p_{i^*} \rceil$. We delete all cheap 
items from $\calI$ and add an item corresponding to $r$ copies of $(p_{i^*}, w_{i^*})$, i.e., we set 
$\calI' := \calI_E \cup \{(r \cdot p_{i^*}, r \cdot w_{i^*})\}$. Note that now any item in $\calI'$ has profit at 
least $T$. These steps are summarized in~\Cref{alg:preprocessing-1}.
The following lemma shows that this decreases the total profit of any solution by at most $O(\varepsilon \cdot \OPT)$.

\begin{algorithm}
    \caption{$\textsc{Preprocessing-Profits}(\calI, W)$: Preprocessing to remove low profit items.}\label{alg:preprocessing-1}
    \begin{algorithmic}[1]
        \State Let $P_0$ be the value of the greedy 2-approximation of $(\calI, W)$
        \State $T := 2 \cdot \varepsilon \cdot P_0$
        \State $\calI_E := \{(p,w) \in \calI \colon p > T\}, \calI_C := \calI \setminus \calI_E$
        \State $i^* := \argmax_{i \in \calI_C} \frac{p_i}{w_i}$
        \State $r := \lceil T / p_{i^*} \rceil$
        \State \Return $\calI_E \cup \{(r \cdot p_{i^*}, r \cdot w_{i^*})\}$
    \end{algorithmic}
\end{algorithm}

\begin{restatable}[Preprocessing Profits]{lemma}{preprocessingprofits} \label{lem:preprocessing-profits}
    Let $\calI'$ be the result of running~\Cref{alg:preprocessing-1} on $(\calI, W)$. Then, for every 
    $w \in [W]$, it holds that $\calP_{\calI'}[w] \geq \calP_{\calI}[w] - 4\varepsilon \cdot \OPT$.
    Moreover, the minimum profit in $\calI'$ is $p_{\min} > \varepsilon \OPT$.
\end{restatable}
\begin{proof}
    Fix some weight $w \in [W]$ and let $x \in \N^{|\calI|}$ be the solution from $\calI$
    corresponding to the entry $\calP_{\calI}[w]$.  We denote by $x_C$ the \emph{cheap} items 
    in $x$, i.e., those with profit  $\leq T = 2 \varepsilon P_0$. Thus, we can write 
    $w_{\calI}(x) = w_{\calI}(x_C) + w_{\calI}(x - x_C)$ and similarly $p_{\calI}(x) = p_{\calI}(x_C) + p_{\calI}(x - x_C)$.

    Now we construct a solution $x' \in \N^{|\calI'|}$ from the preprocessed $\calI'$ which will satisfy 
    $w_{\calI'}(x') \leq w$ and $p_{\calI'}(x') \geq p - \varepsilon \cdot \OPT$, yielding the lemma.
    We start from $x$ and remove all the cheap items $x_C$, replacing them by $\lfloor w_{\calI}(x_C) / (r \cdot w_{i^*}) \rfloor$-many copies of 
    $(r \cdot p_{i^*}, r \cdot w_{i^*}) \in \calI'$. Note that in this way, $w_{\calI'}(x') \leq w$. Now we lower bound the profit.
    Since $(p_{i^*}, w_{i^*})$ is the item corresponding to the greedy 2-approximation from the cheap items, it follows that 
    $p_{\calI}(x_C) \leq  \frac{w_{\calI}(x_C)}{w_{i^*}} \cdot p_{i^*}$. Further, by definition of $r$ it holds that $r \leq T / p_{i^*} + 1$. 
    Thus, 
    \[
        p_{\calI'}(x') 
            = p_{\calI}(x - x_C) + r \cdot p_{i^*} \cdot \lfloor w(x_C) / (r \cdot w_{i^*}) \rfloor
            \geq p_{\calI}(x - x_C) + p_{\calI}(x_C) - r \cdot p_{i^*}
            \geq p_{\calI}(x) - T - p_{i^*}.
    \]
    Since $p_{i^*} \leq T = 2 \varepsilon \cdot P_0$, we conclude that $p_{\calI'}(x') \geq p_{\calI}(x) - 4\varepsilon \cdot P_0$.
    Therefore, $\calP_{\calI'}[w] \geq p_{\calI'}(x') \geq p_{\calI}(x) - 4\varepsilon \cdot P_0 \geq \calP_{\calI}[w] - 4\varepsilon \cdot \OPT$, as desired.

    Finally, note that by construction we have that $p_{\min} > 2 \varepsilon P_0 \geq \varepsilon \OPT$.
\end{proof}

\paragraph*{Step 2: Remove light items}

Having removed all items with low profit, we proceed similarly to remove items of low weight.
More precisely, we say that an item is light if its weight is less than $\varepsilon \cdot W$.
We remove all light items except for the most profitable one (i.e. the one with the best profit to weight ratio), 
which we \emph{copy} enough times to make it have weight at least $\varepsilon \cdot W$. The details 
are shown in~\Cref{alg:preprocessing-2}.

\begin{algorithm}
    \caption{$\textsc{Preprocessing-Profits-And-Weights}(\calI, W)$: Preprocessing to remove items of low profit and low weight.}\label{alg:preprocessing-2}
    \begin{algorithmic}[1]
        \State $\widetilde{\calI} := \textsc{Preprocessing-Profits}(\calI, W)$ (\Cref{alg:preprocessing-1})
        \State $\calI_L := \{(p,w) \in \widetilde{\calI} \colon w < \varepsilon \cdot W\}, \calI_H := \widetilde{\calI} \setminus \calI_S$
        \State $i^* := \argmax_{i \in \calI_L} \frac{p_i}{w_i}$
        \State $r := \lceil \frac{\varepsilon \cdot W} {w_{i^*}} \rceil$
        \State \Return $\calI_H \cup \{(r \cdot p_{i^*}, r \cdot w_{i^*})\}$
    \end{algorithmic}
\end{algorithm}

The following lemma shows that removing cheap and light items as in~\Cref{alg:preprocessing-2} only decreases the profit of 
any solution by $O(\varepsilon \OPT)$.
\begin{restatable}[Preprocessing Profits and Weights]{lemma}{preprocessingprofitsweights}\label{lem:preprocessing-profits-and-weights}
    Let $\calI'$ be the result of running~\Cref{alg:preprocessing-2} on $(\calI, W)$. 
    Then any solution for $\calI'$ has a corresponding solution for $\calI$ with the same profit and weight.
    Further, for any $w \in [W]$ it holds that $\calP_{\calI'}[w] \geq \calP_{\calI}[w] - 8\varepsilon \cdot \OPT$.
    Moreover, the minimum profit and minimum weight in $\calI'$ satisfy $p_{\min} > \varepsilon\OPT$ and $w_{\min} > \varepsilon W$.
\end{restatable}

\begin{proof}
    Fix $w \in [W]$ and $p := \calP_{\calI}[w]$. 
    After running~\Cref{alg:preprocessing-1} in Line 1, by~\Cref{lem:preprocessing-profits}, we know that there is a solution $x$ from $\widetilde{\calI}$ with 
    $p_{\widetilde{\calI}}(x) \geq p - 4\varepsilon \cdot \OPT$ and $w_{\widetilde{\calI}}(x) \leq w$. Let $x_L$ be the light items in $x$, i.e., 
    those with weight $< \varepsilon \cdot W$. 
    We construct a solution $x'$ of items from $\calI'$ by taking $x$, removing all the items in $x_L$ 
    and replacing them by $\lfloor w_{\widetilde{\calI}}(x_L) / (r \cdot w_{i^*})  \rfloor$-many copies of 
    $(r \cdot p_{i^*}, r \cdot w_{i^*})$. Then, the total profit of $x'$ is

    \begin{align}
        p_{\calI'}(x') &= p_{\widetilde{\calI}}(x - x_L) 
            +  \left\lfloor \frac{w_{\widetilde{\calI}}(x_L)}{r \cdot w_{i^*}} \right\rfloor\cdot r \cdot p_{i^*} 
            \geq p_{\widetilde{\calI}}(x - x_L) 
            +  \frac{w_{\widetilde{\calI}}(x_L)}{r \cdot w_{i^*}} \cdot r \cdot p_{i^*} - r \cdot p_{i^*} \nonumber \\
            & \geq p_{\widetilde{\calI}}(x - x_L) + p_{\widetilde{\calI}}(x_L)  - r\cdot p_{i^*} = p_{\widetilde{\calI}}(x) - r \cdot p_{i^*} \label{eqn:bound-profit},
    \end{align}
    where the second inequality $p(x_L) \leq \frac{w(x_L)}{r \cdot w_{i^*}} \cdot r \cdot p_{i^*}$ follows since by definition 
    $(p_{i^*}, w_{i^*})$ is the most efficient of the light items. Now we argue that we can bound the last term
    $p_{i^*} \cdot r$ by  $O(\varepsilon \OPT)$. As $w_{i^*}$ is a light item, we have $w_{i^*} \leq \varepsilon W$.
    Since for any number $z \geq 1$ we have that $\lceil z  \rceil \leq 2 z$, it follows that 
    $r = \lceil (\varepsilon W) / w_{i^*} \rceil \leq 2 \varepsilon W / w_{i^*}$.
    Since packing $\lfloor W / w_{i^*} \rfloor$ copies of item $i^*$ is a feasible solution, it holds that 
    $\OPT \geq \lfloor W / w_{i^*} \rfloor p_{i^*} \geq \tfrac{W p_{i^*}}{2 w_{i^*}}$. Combining 
    both inequalities we get that $r \cdot p_{i^*} \leq 2 \varepsilon W p_{i^*} / w_{i^*} \leq 4 \varepsilon \cdot \OPT$.
    Therefore, by~\eqref{eqn:bound-profit} we conclude that 
    $p_{\calI'}(x') \geq p_{\widetilde{\calI}}(x) - r \cdot p_{i^*} \geq p_{\widetilde{\calI}}(x) - 4\varepsilon \cdot \OPT \geq p - 8\varepsilon \cdot \OPT$.
    
    Similarly, we can bound the total weight of $x'$ as
    \[
       w_{\calI'}(x') = w_{\widetilde{\calI}}(x - x_L) + \left\lfloor \frac{w_{\widetilde{\calI}}(x_L)}{r \cdot w_{i^*}} \right\rfloor \cdot r \cdot w_{i^*} 
       \leq w_{\widetilde{\calI}}(x - x_L) + w_{\widetilde{\calI}}(x_L)
       \leq w_{\widetilde{\calI}}(x) \leq w.
    \]

    Finally, note that $p_{\min} > \varepsilon \OPT$ follows directly from~\Cref{lem:preprocessing-profits} and 
    $w_{\min} > \varepsilon W$ by construction.
\end{proof}

\subsubsection{The Algorithm}

To obtain the FPTAS, we first preprocess the set of items $\calI$ using~\Cref{alg:preprocessing-2} and obtain a set $\calI'$. 
Then, by~\Cref{lem:preprocessing-profits-and-weights}, it holds that $\calP_{\calI'}[W] \geq (1 - O(\varepsilon))\cdot \OPT$,
so to obtain a $(1 + O(\varepsilon))$-approximation of $\calP_{\calI}[W]$ it suffices to give a $(1+\varepsilon)$-approximation of $\calP_{\calI'}[W]$.

In fact, we solve a slightly more general problem, and compute a sequence which gives a (pointwise) $(1+\varepsilon)$-approximation 
of the entire sequence $\calP_{\calI'}[0,\dots,W]$, according to the following definition.

\begin{definition}[Strong Approximation of a Sequence]\label{def:approx-sequence}
    We say that a sequence $A[0,\dots,n]$ (pointwise) $(1+\varepsilon)$-approximates a sequence $B[0,\dots,n]$ if
    $B[i]/(1+\varepsilon) \leq A[i] \leq B[i]$ for every $i \in [n]$.
\end{definition}

Our main ingredient will be the following approximate 
$(\max,+)$-convolution from~\cite{Chan18a} (the proof is very similar to that of~\Cref{lem:max-conv-weak-approx}, which we give in the next section).

\begin{lemma}[{\cite[{Item (i) from Lemma 1}]{Chan18a}}]\label{lem:max-conv-approx}
    Let $A[0,\dots,n], B[0,\dots,n]$ be two monotone non-decreasing sequences with $t_1$ and $t_2$ steps, respectively.
    Let $p_{\max}, p_{\min}$ be the maximum and minimum non-zero values in $A, B$. Then, for any $\varepsilon \in (0,1)$
    a monotone non-decreasing sequence $C$ with $O(1/\varepsilon \cdot \log(p_{\max}/ p_{\min}))$ 
    steps which gives a (pointwise) $1 + O(\varepsilon)$-approximation of $A \oplus B$ can be computed in time
    \[
        O((t_1 + t_2 + T(1/\varepsilon)) \cdot \log(p_{\max} / p_{\min}))
    \] 
    where $T(n)$ is the time needed to compute the $(\min,+)$-convolution of two length-$n$ sequences.
\end{lemma}

Now we are ready to describe the main routine. Since after the preprocessing we have that $w_{\min} > \varepsilon W$,
any feasible solution contains at most $W / w_{\min} < 1/\varepsilon$ items. Thus, to approximate the optimal 
solution in $\calI'$, we apply~\Cref{lem:splitting-lemma} for $O(\log 1/\varepsilon)$ iterations, but replace 
the $(\max,+)$-convolutions with the approximate ones from~\Cref{lem:max-conv-approx}. The pseudocode is 
written in~\Cref{alg:fptas}.

\begin{algorithm}
    \caption{$\textsc{FPTAS}(\calI, W)$: Given an instance $(\calI, W)$ of \UnboundedKnapsack{} with $w_{\min} > \varepsilon W$, the algorithm 
    returns a $(1 + O(\varepsilon))$-approximation of $\calP_{\calI}[0,\dots,W]$}\label{alg:fptas}
    \begin{algorithmic}[1]
        \State Initialize $S_0 := \calP_0[0,\dots,W]$, stored implicitly
        \For{$i = 1, \dots,\lceil \log 1 / \varepsilon \rceil$}
            \State Approximate $S_i := (S_{i-1} \oplus S_{i-1})[0,\dots,W]$
            \Statex[1] using~\Cref{lem:max-conv-approx} with error parameter 
                $\varepsilon' := \varepsilon/\log(1/\varepsilon)$\label{alg:fptas:line:conv-approx}
        \EndFor
        \State \Return $S_{\lceil \log 1 / \varepsilon \rceil}[0,\dots,W]$
    \end{algorithmic}
\end{algorithm}

\begin{lemma}\label{lem:guarantee-fptas}
    Let $(\calI, W)$ be an instance of \UnboundedKnapsack{} on $n$ items with $p_{\min} > \varepsilon \OPT$ and $w_{\min} > \varepsilon W$.
    Then, on input $(\calI, W)$ \Cref{alg:fptas} computes a pointwise $(1 + O(\varepsilon))$-approximation of $\calP_{\calI}[0,\dots,W]$
    in time $\widetilde{O}(n + T(1/\varepsilon))$, where $T(n)$ is the time to compute \BMMaxConv{} on length-$n$ sequences.
\end{lemma}
\begin{proof}
    First we argue about correctness. We claim that if the $(\max,+)$-convolutions of \Cref{alg:fptas:line:conv-approx} were \emph{exact},
    then for each $i$ we have that $S_i = \calP_{i}[0,\dots,W]$. For the base case $i = 0$, this holds by definition. For further 
    iterations, this holds inductively by~\Cref{lem:splitting-lemma}.
    Since we replaced the exact computations with approximate ones with error parameter $\varepsilon' = \varepsilon/\log(1/\varepsilon)$, 
    and since there are $O(\log 1/\varepsilon)$ iterations where the error
    accumulates multiplicatively, we conclude that $S_{\lceil \log 1/\varepsilon \rceil}[0,\dots,W]$ is a pointwise 
    $(1+\varepsilon')^{O(\log (1/\varepsilon))} = 1 + O(\varepsilon' \log(1/\varepsilon)) = (1 + O(\varepsilon))$-approximation 
    of $\calP_{\lceil \log 1/\varepsilon \rceil}[0,\dots, W]$. Finally note that since we assume that $w_{\min} > \varepsilon W$,
    any feasible solution consists of at most $W/w_{\min} < 1/\varepsilon$ items, which implies that 
    $\calP_{\lceil \log 1/\varepsilon \rceil}[0,\dots,W] = \calP[0,\dots,W]$.

    Now we argue about the running time. The initialization in line 1 takes time $O(n)$, since by definition 
    $\calP_0[0,\dots,W]$ has $n$ steps corresponding to the $n$ solutions containing one item.
    Next, we look at the for-loop. For each iteration, we use~\Cref{lem:max-conv-approx} to approximate the convolutions.
    Note that each application involves sequences with $O(n + 1/\varepsilon \log(p^*/p_{\min}))$ steps by the 
    guarantee of~\Cref{lem:max-conv-approx} (the $n$ term arising from the steps of $S_0$) where $p^*$ is the maximum entry in any of the involved sequences.
    Since each entry in any of the sequences corresponds to a feasible 
    solution for the instance, we have that $p^* \leq \OPT$. Moreover by assumption we have that $p_{\min} \ge \varepsilon \OPT$,
    and hence $p^*/p_{\min} \leq O(1/\varepsilon)$. Therefore, each 
    iteration takes time $\widetilde{O}(n + 1/\varepsilon + T(1/\varepsilon))$ and the overall time is bounded by
    \[
        \sum_{i = 1}^{\lceil \log \frac{1}{\varepsilon} \rceil} \widetilde{O}\left(n + \tfrac{1}{\varepsilon} + T(1/\varepsilon)\right) 
            = \widetilde{O}(n + T(1/\varepsilon)).
    \]
    Note that in this proof we are using the niceness assumption that $T(\widetilde{O}(n)) = \widetilde{O}(T(n))$ as mentioned 
    in~\Cref{sec:niceness}.
\end{proof}

We put the pieces together to conclude the proof of the main theorem of this section.
\begin{proof}[Proof of~\Cref{thm:reduction-fptas-to-minconv}]
    Given a instance of $(\calI, W)$ of \UnboundedKnapsack{}, we first run the preprocessing from 
    \Cref{alg:preprocessing-2} and obtain a new instance $(\calI', W)$. Then, we run~\Cref{alg:fptas} on 
    $(\calI', W)$ and obtain a sequence $S[0,\dots,W]$.
    We claim that $S[W]$ gives the desired $(1 + O(\varepsilon))$-approximation to $\OPT$ in the desired running time.

    To see the running time, note that after the preprocessing it holds that 
    $p_{\min} > 2 \varepsilon P_0 \geq \varepsilon \OPT$ and $w_{\min} > \varepsilon W$.
    Hence, we can apply~\Cref{lem:guarantee-fptas} and conclude that the procedure runs in the desired running time.
    To argue about correctness, note that by~\Cref{lem:preprocessing-profits-and-weights} it holds that 
    $\calP_{\calI'}[W] \geq (1 - 8\varepsilon)\OPT$. Combining this with the fact that $S[0,\dots,W]$ $(1 + O(\varepsilon))$-approximates 
    $\calP_{\calI'}[0,\dots,W]$ due to~\Cref{lem:guarantee-fptas},  we conclude that
    \[
        S[W] \geq \frac{\calP_{\calI'}[W]}{1 + O(\varepsilon)}
        \geq \frac{(1 - 8\varepsilon)\OPT}{1 + O(\varepsilon)} \geq \frac{\OPT}{1 + O(\varepsilon)},
    \] as desired.
\end{proof}

\subsection{Faster Weak FPTAS}\label{sec:weak-fptas}

In this section we prove the following theorem.

\begin{theorem}\label{thm:reduction-weakfptas-to-minconv}
    If \BMMaxConv{} can be solved in time $T(n)$, then \UnboundedKnapsack{} has a 
    weak approximation scheme in time $\widetilde{O}(n + T(1/\varepsilon))$.
\end{theorem}

Note that~\Cref{thm:weak-fptas} follows as an immediate corollary by plugging in
Chi, Duan, Xie and Zhang's algorithm (\Cref{thm:convolution-bounded}).

\subsubsection{The Algorithm}

Our approach is similar to the FPTAS from the previous section. The main difference is that since now we can overshoot the 
weight constraint by an additive term of $\varepsilon W$, we can afford to round the weights of the items. By doing so, 
we can adapt the algorithm from~\Cref{lem:max-conv-approx} to use an algorithm for \BMMaxConv{}. 

We start by specifying the notion of approximation for monotone sequences we will be working with. 
In the following definition, it is helpful to think of the indices as weights and the entries as profits.
\begin{definition}[Weak Approximation of a Sequence]\label{defn:weak-approx-sequence}
    We say that a sequence $A[0,\dots,n']$ \emph{weakly} $(1+\varepsilon)$-approximates a sequence $B[0,\dots,n]$ if 
    the following two properties hold:
    \begin{itemize}
        \item[(i)] (Good approximation for every entry) For every $j \in [n]$, there exists $j' \in [n']$ with $A[j'] \geq B[j]/(1 + \varepsilon)$ and 
            $j' \leq (1 + \varepsilon)j$.
        \item[(ii)] (No approximate entry is better than a true entry) For every $j' \in [n']$, there exists $j \in [n]$ such that 
        $A[j'] \leq B[j]$ and $j' \geq j$.
    \end{itemize}
\end{definition}

The following lemma shows that our notion of weak approximations is transitive. In particular, it implies 
that a weak approximation of the $(\max,+)$-convolution of two weakly approximated sequences gives a (slightly worse)
weak approximation of the convolution of the original sequences.

\begin{lemma}\label{lem:error-propagation}
    If $A$ weakly $(1+\varepsilon)$-approximates $B$ and $B$ weakly $(1 + \varepsilon')$-approximates $C$, 
    then $A$ weakly $(1 + \varepsilon)(1 + \varepsilon')$-approximates $C$.
\end{lemma}
\begin{proof}
    First we check that $A$ and $C$ satisfy property (i) of \Cref{defn:weak-approx-sequence}. Fix an entry $C[k]$. Since $B$ 
    weakly $(1 + \varepsilon')$-approximates $C$, by property (i) there exists an index $j$ such that $B[j] \geq C[k]/(1+\varepsilon')$
    and $j \leq (1+\varepsilon')k$. Thus, since $A$ weakly $(1+\varepsilon)$-approximates $B$ we apply property (i) once more to 
    conclude that there is an index $i$ such that $A[i] \geq B[j]/(1+\varepsilon) \geq \tfrac{C[k]}{(1+\varepsilon)(1+\varepsilon')}$
    and $i \leq (1+\varepsilon)j \leq (1+\varepsilon)(1+\varepsilon') k$. 

    We can similarly check that $A$ and $C$ satisfy property (ii). Fix an entry $A[i]$ and apply property (ii) to obtain an 
    index $j$ such that $A[i] \leq B[j]$ and $i \geq j$. Applying property (ii) once more to $j$ yields an index $k$ such that 
    $B[j] \leq C[k]$ and $j \geq k$. Combining, we conclude that $A[i] \leq C[k]$ and $i \geq k$, as desired.
\end{proof}

We extend~\Cref{lem:max-conv-approx} to give a weak approximation of the $(\max,+)$-convolution of two sequences.
Given two pairs of integers $(w,p),(w',p')$,
we say that $(w,p)$ \emph{dominates} $(w',p')$ if $w \leq w'$ and $p > p'$. Given a list of pairs $D = [(w_1,p_1),\dots,(w_m, p_m)]$, 
we can remove all dominated pairs from $D$ in near-linear time for example by sorting.

\begin{lemma}\label{lem:max-conv-weak-approx}
    Let $A[0,\dots,N], B[0,\dots,N]$ be two monotone non-decreasing sequences with $t_1$ and $t_2$ steps, respectively.
    Let $p_{\max}, p_{\min}$ be the maximum and minimum non-zero values in $A, B$ and let 
    $w_{\min} = \min\{i \mid A[i] \neq 0 \text{ or } B[i] \neq 0\}$. Then, for any $\varepsilon \in (0,1)$ a monotone non-decreasing sequence $C$ 
    which is a weak $1 + O(\varepsilon)$-approximation of $(A \oplus B)[0,\dots,2N]$ and consists of 
    $O(1/\varepsilon \cdot \log(p_{\max}/p_{\min}) \cdot \log(N/w_{\min}))$ steps can be computed in time
    \[
        O((t_1 + t_2 + T(1/\varepsilon)) \cdot \log(p_{\max}/p_{\min}) \log(N/w_{\min}))
    \]
    where $T(n)$ is the time needed to compute \BMMaxConv{} on sequences of length $n$.
\end{lemma}
\begin{proof}
    Let $C := A \oplus B$. For this proof it will be convenient to work directly with the steps of the monotone 
    sequences $A, B$ and $C$. So suppose we are given $A$ as a list of steps $A = [(w_1, p_1), \dots, (w_{t_1}, p_{t_1})]$
    where each step $(w, p)$ represents an entry $p = A[w]$ at which $A$ changes, i.e., $A[w-1] < A[w]$ (or $w = 0$).  
    Similarly, $B$ and $C$ are represented as lists of steps.

    Fix integers $p^*, w^* > 0$ and define a sequence of steps $A'$ by keeping every step $(w,p) \in A$ with $w \leq w^*$ and $p \leq p^*$
    and rounding it to $(\lceil w/(\varepsilon w^*) \rceil, \lfloor p / (\varepsilon p^*) \rfloor)$.
    Let $B'$ be defined analogously.
    Compute $C' := A' \oplus B'$, and scale each step $(w', p') \in C'$ back, i.e., 
    set $w := w' \cdot (\varepsilon w^*)$ and $p := p \cdot (\varepsilon p^*)$, obtaining a sequence of steps $C''$. We claim that:

    \begin{itemize}
        \item[(a)] For every step $(w, p) \in C$ for which $w^*/2 \leq w \leq w^*$ and $p^*/2 \leq p \leq p^*$,
        there is a step $(w'', p'') \in C''$ such that $w'' \leq (1 +4\varepsilon)w$ and 
        $p'' \geq p/(1 + 8\varepsilon)$.
        \item[(b)] For every step $(w'', p'') \in C''$ there is a step 
        $(w, p) \in C$ such that $w'' \geq w$ and $p'' \leq p$.
    \end{itemize}

    To see (a), fix one such step $(w,p) \in C$ and let $(w_a, p_a) \in A, (w_b, p_b) \in B$ be the corresponding steps 
    such that $w = w_a + w_b$ and $p = p_a + p_b$. Note that $(w_a, p_a)$ has a corresponding rounded step 
    $(w_a', p_a') \in A'$ since $w \leq w^*$ and $p \leq p^*$ imply that 
    $(w_a, p_a)$ were included (after rounding) to $A'$. Similarly, $(w_b, p_b)$ has a corresponding step $(w_b', p_b') \in B'$. Hence $C'[w_a' + w_b'] \geq p_a' + p_b'$,
    and therefore there is a step $(w_c'', p_c'') \in C''$ such that 
    \[
        w_c'' \leq (w_a' + w_b') \cdot \varepsilon w^* \leq (w/(\varepsilon  w^*) + 2) \cdot \varepsilon w^* \leq (1 + 4\varepsilon)w,    
    \]
    where we used the definition of $w'_a, w'_b$ in the second inequality, and $w^*/2 \leq w$ in the last one.
    Similarly, we have that
    \[
        p_c'' 
            \geq (p_a' + p_b') \cdot \varepsilon p^* 
            \geq (p/(\varepsilon p^*) - 2) \cdot \varepsilon p^* 
            \geq (1 - 4\varepsilon)p
            \geq p / (1 + 8\varepsilon),
    \]
    where in the last step we assumed $\varepsilon \le \frac 18$, which is without loss of generality. This proves~(a). 
    
    Part (b) follows because each step $(w'', p'') \in C''$ is a (scaled) sum of steps $(w_a', p_a') \in A'$ and 
    $(w_b', p_b') \in B'$ which by construction are rounded steps $(w_a, p_a) \in A, (w_b, p_b) \in B$. Thus,
    \begin{align*}
        w'' &= (w_a' + w_b') \cdot \varepsilon w^* 
            = (\lceil w_a/(\varepsilon w^*) \rceil + \lceil w_b/(\varepsilon w^*) \rceil) \cdot \varepsilon w^*
            \geq w_a + w_b, \\
        p'' &= (p_a' + p_b') \cdot \varepsilon p^* 
        = (\lfloor p_a/(\varepsilon p^*) \rfloor + \lfloor p_b/(\varepsilon p^*) \rfloor) \cdot \varepsilon p^*
        \leq p_a + p_b,
    \end{align*}
    which implies that there is a step $(w, p) \in C$ such that $w \leq w_a + w_b \leq w''$ and 
    $p \geq p_a + p_b \geq p''$.

	\smallskip
    The claim implies that if we apply the described procedure for every $w_{\min} \leq w^* \leq 2N$ and $p_{\min} \leq p^* \leq p_{\max}$
    which are powers of 2 and combine the results by keeping the set of non-dominated steps, we obtain a weak 
    $(1 + O(\varepsilon))$-approximation of $C$.
    Indeed, part (a) implies property (i) of \Cref{defn:weak-approx-sequence} and part (b)
    implies property (ii).
    The overall procedure is summarized in~\Cref{alg:weak-approx-convolution}.

    Now we analyze the running time. For each $p^*$ and $w^*$, we spend $O(t_1 + t_2)$ time to construct $A'$ and $B'$ in lines 6 and 7. 
    The key step is the computation of $C'$ in line 9. Note that the steps constructed in line 6 for $A'$ define a monotone 
    sequence $A'[1,\dots,\lfloor 1/\varepsilon \rfloor]$ where $A'[i] := \max\{p' \colon (w', p') \in A', w' \leq i\}$, 
    and $B'[1,\dots,\lfloor 1/\varepsilon \rfloor]$ is defined similarly. Thus, $A', B'$ are monotone non-decreasing 
    sequences of length $O(1/\varepsilon)$, and due to the rounding their entries are bounded by $O(1/\varepsilon)$. Thus, we can compute 
    $C'$ in time $T(1/\varepsilon)$ and the overall running time of the algorithm is
    $O((t_1 + t_2 + T(1/\varepsilon)) \cdot \log(p_{\max}/{p_{\min}})\log(N/w_{\min}))$.
\end{proof}

\begin{algorithm}
    \caption{Given monotone non-decreasing sequences $A, B$ represented as lists of steps, the algorithm 
    returns the steps of a weak $(1 + O(\varepsilon))$-approximation of $A \oplus B$.}\label{alg:weak-approx-convolution}
    \begin{algorithmic}[1]
        \State $S := [ \;]$
        \For{$i = 0, 1, 2, \dots, \lceil \log (p_{\max} / p_{\min}) \rceil$}
            \State $p^* := 2^i \cdot p_{\min}$
            \For{$j = 0, 1, 2, \dots, \lceil \log (N/w_{\min}) \rceil$}
                \State $w^* := 2^j \cdot w_{\min}$
                \State $A' := \big[
                    (\lceil \tfrac{w}{\varepsilon \cdot w^*} \rceil, 
                    \lfloor \tfrac{p}{\varepsilon \cdot p^*} \rfloor)
                    \mid (w,p) \in A, \; w \leq w^*, p \leq p^*\big]$
                \State $B' := \big[
                    (\lceil \tfrac{w}{\varepsilon \cdot w^*} \rceil,
                    \lfloor \tfrac{p}{\varepsilon \cdot p^*} \rfloor) 
                    \mid (w,p) \in B, \; w \leq w^*, p \leq p^*\big]$
                \State Remove dominated pairs from $A'$ and $B'$
            \State Compute $C' := A' \oplus B'$ using~\Cref{thm:convolution-bounded}
                \State $C'' := [(w \cdot \varepsilon w^*, p \cdot \varepsilon p^*) \mid (w,p) \in C']$
                \State Append $C''$ to $S$
            \EndFor
        \EndFor
        \State Remove dominated pairs in $S$ and sort the remaining points
        \State \Return $S$
    \end{algorithmic}
\end{algorithm}

\bigskip

Now we are ready to give the main algorithm. The idea is the same as in the FPTAS in~\Cref{alg:fptas},
except that we replace the \MaxConv{} computations, i.e., we use~\Cref{lem:max-conv-weak-approx} instead of~\Cref{lem:max-conv-approx}.
We will use the following notation. Given a sequence $A[0,\dots,n]$, and an integer $i \geq 1$ we denote by 
$(A[0,\dots,n])^i$ the result of applying $2^i$ $(\max,+)$-convolutions of $A$ with itself, i.e.
\[
    (A[0,\dots,n])^i := \underbrace{A[0,\dots,n] \oplus \dots \oplus A[0,\dots,n]}_{2^i\text{ times}}.
\]
If $i = 0$, we set $(A[0,\dots,n])^0 := A[0,\dots,n]$. Note that the length of the resulting sequence is $2^i \cdot n$.

\begin{algorithm}
    \caption{$\textsc{Weak-FPTAS}(\calI, W)$: Given an instance $(\calI, W)$ of \UnboundedKnapsack{}, 
    the algorithm returns a weak $(1 + O(\varepsilon))$-approximation of $(\calP_{\calI,0}[0,\dots,W])^{\lceil \log 1/\varepsilon \rceil}$}\label{alg:weak-fptas}
    \begin{algorithmic}[1]
        \State Initialize $S_0 := \calP_0[0,\dots,W]$, stored implicitly
        \For{$i = 1, \dots,\lceil \log 1/\varepsilon \rceil$}
            \State Approximate 
            $S_i := S_{i-1} \oplus S_{i-1}$ using~\Cref{lem:max-conv-weak-approx} with error parameter $\varepsilon' := \varepsilon/\log(1/\varepsilon)$
        \EndFor
        \State \Return $S_{\lceil \log 1 / \varepsilon \rceil}$
    \end{algorithmic}
\end{algorithm}

\begin{lemma}\label{lem:guarantee-weak-fptas}
    Let $(\calI, W)$ be an instance of \UnboundedKnapsack{} on $n$ items with $p_{\min} > \varepsilon \OPT$ and $w_{\min} > \varepsilon W$.
    Then, on input $(\calI, W)$~\Cref{alg:weak-fptas} computes a weak
    $(1 + O(\varepsilon))$-approximation of $(\calP_{0}[0,\dots,W])^{\lceil \log 1/\varepsilon \rceil}$ in time $\widetilde{O}(n + T(1/\varepsilon))$, where $T(n)$ is the time needed to compute \BMMaxConv{} on length-$n$ sequences.
\end{lemma}
\begin{proof}
    First we argue correctness. We will show by induction that for every $i \geq 0$ it holds that $S_i$ is a weak 
    $(1 + \varepsilon')^i$-approximation of $(\calP_0[0\dots,W])^i$. For the base case $i = 0$, this holds by definition of $S_0$. 
    For the inductive step, take $i > 0$ and assume 
    the claim holds for $i - 1$. By the computation in line 3 of \Cref{alg:weak-fptas}, 
    $S_i$ weakly $(1 + \varepsilon')$-approximates $S_{i-1} \oplus S_{i-1}$. And by the inductive 
    hypothesis, $S_{i-1}$ is a weak $(1 + \varepsilon')^{i-1}$-approximation of $(\calP_0[0,\dots,W])^{i-1}$.
    Therefore, by~\Cref{lem:error-propagation} it follows that $S_i$ weakly $(1+\varepsilon')^i$-approximates 
    $(\calP_0[0,\dots,W])^{i-1} \oplus (\calP_0[0,\dots,W])^{i-1} = (\calP_0[0,\dots,W])^{i}$.
    Since $\varepsilon' = \varepsilon/\log(1/\varepsilon)$, we conclude that $S_{\lceil \log 1/\varepsilon \rceil}$
    is a weak $(1 + \varepsilon')^i = (1 + O(\varepsilon))$-approximation of $(\calP_{0}[0,\dots,W])^{\lceil \log 1/\varepsilon \rceil}$, as desired.

    Now we analyze the running time. The initialization in line 1 takes time $O(n)$. 
    Then, in each iteration of the for-loop we apply~\Cref{lem:max-conv-weak-approx} to approximate the 
    $(\max,+)$-convolutions. When $i = 1$ the input sequences have $O(n)$ steps, and we obtain a sequence 
    with $O(1/\varepsilon \log(p_{\max}/p_{\min}) \log(W/w_{\min}))$ steps which approximates $\calP_0[0,\dots,W] \oplus \calP[0,\dots,W]$.
    Continuing inductively, in the $i$-th iteration we approximate the $(\max, +)$-convolution of sequences with
    $O(n + 1/\varepsilon \log(2^i p_{\max}/p_{\min}) \log(2^i W/w_{\min}))$ steps. 
    Since $i \leq O(\log 1 /\varepsilon)$, and by assumption $p_{\min} > \varepsilon \OPT$ and $w_{\min} > \varepsilon W$
    it follows that the number of steps of the input sequences in every iteration are bounded by $\widetilde{O}(n + 1/\varepsilon)$.
    Therefore, each iteration takes time $\widetilde{O}(n + T(1/\varepsilon))$ and thus we can bound the overall running time by 
    \[
        \sum_{i = 1}^{\lceil \log \frac{1}{\varepsilon} \rceil} \widetilde{O}\left(n + \tfrac{1}{\varepsilon} + T(1/\varepsilon)\right) 
            = \widetilde{O}(n + T(1/\varepsilon)).
    \]
    Note that here we are using the niceness assumption that $T(\widetilde{O}(n)) = \widetilde{O}(T(n))$ as mentioned 
    in~\Cref{sec:niceness}.
\end{proof}

Now we can put things together to prove the main theorem.

\begin{proof}[Proof of~\Cref{thm:reduction-weakfptas-to-minconv}]
    Given a instance of $(\calI, W)$ of \UnboundedKnapsack{}, we first preprocess it with 
    \Cref{alg:preprocessing-2} and obtain a new instance $(\calI', W)$. Then, we run~\Cref{alg:weak-fptas} on 
    $(\calI', W)$ and obtain a sequence $S$. We claim that $\tilde{p} := S[\ell]$ where $\ell := (1 + O(\varepsilon)) W$ 
    (for a sufficiently large hidden constant given by the guarantee of~\Cref{lem:guarantee-weak-fptas}) 
    is the profit of a feasible solution with profit $\tilde{p} \geq \OPT/(1+O(\varepsilon))$ and weight at most $(1 + O(\varepsilon))W$,
    and that we can compute it in the desired running time.

    To see the claim, first note that after the preprocessing, it holds that 
    $p_{\min} > 2 \varepsilon P_0 \geq \varepsilon \OPT$ and $w_{\min} > \varepsilon W$.
    Hence, we can apply~\Cref{lem:guarantee-weak-fptas} and conclude that we can compute $S$ in the desired running time.
    To argue about correctness, let $x^*$ be the optimal solution for the original instance $(\calI, W)$, i.e., 
    $\calP_{\calI}[W] = p(x^*) = \OPT$.
    By~\Cref{lem:preprocessing-profits-and-weights}, there is a solution $\tilde{x}$ from $\calI'$ which satisfies 
    $p(\tilde{x}) \geq (1 - 8\varepsilon)p(x^*) = (1-8\varepsilon)\OPT$ and $w(\tilde{x}) \leq w(x^*) \leq W$. Thus,
    it suffices to argue that $S[\ell]$ corresponds to a solution from $\calI'$ which weakly-approximates $\tilde{x}$.
    
    Since the maximum number of items in any feasible solution for $(\calI', W)$ is at most $W/w_{\min} < 1/\varepsilon$, 
    we can apply \Cref{lem:splitting-lemma} inductively to obtain that
    \begin{equation} \label{eqn:squaring-guarantee}
        (\calP_{\calI',0}[0,\dots,W])^{\lceil \log 1/\varepsilon \rceil}[0,\dots,W] = \calP_{\calI'}[0,\dots,W].
    \end{equation}
    Now, \Cref{lem:guarantee-weak-fptas} guarantees that the output sequence $S$ is a 
    weak $(1+O(\varepsilon))$-approximation of $(\calP_{\calI',0}[0,\dots,W])^{\lceil \log 1 / \varepsilon \rceil}$.
    Therefore by combining \eqref{eqn:squaring-guarantee}, property (i) of~\Cref{defn:weak-approx-sequence} and monotonicity
    we have that $S_{\lceil \log 1/\varepsilon \rceil}[\ell] \geq \OPT/(1+O(\varepsilon))$. Moreover, by 
    property (ii) of~\Cref{defn:weak-approx-sequence} we know that this indeed corresponds to a solution from 
    $\calI'$ of weight at most $(1 + O(\varepsilon)) W$.
\end{proof}

\subsection{Solution Reconstruction}

For both our strong and weak approximation schemes presented in \Cref{sec:approx}, we only described how to obtain an approximation of the
\emph{value} of the optimal solution. We can also reconstruct the solution itself by computing the witnesses of each 
$(\max,+)$-convolution (see \Cref{sec:witnesses}), as we show in the following lemma.

\begin{lemma}
    A solution attaining the value given by~\Cref{alg:weak-fptas} can be found in time $\widetilde{O}(n + T(1/\varepsilon))$,
    where $T(n)$ is the time to compute \BMMaxConv{} on length-$n$ sequences.
\end{lemma}
\begin{proof}[Proof Sketch]
    We run~\Cref{alg:weak-fptas}, but additionally compute witnesses for each convolution. More precisely,
    consider the sequence $S_i$ obtained at each iteration $i$ of~\Cref{alg:weak-fptas}. 
    Note that via~\Cref{lem:witness-finding} we can obtain the witness array for every step in $S_i$
    by computing the witnesses of each $(\max,+)$-convolution computed inside~\Cref{lem:max-conv-weak-approx},
    which only adds a polylogarithmic overhead to the running time.
    After doing this, for every step in $S_i$ we can obtain the corresponding two steps in $S_{i-1}$
    that define it in constant time (by doing a lookup in the witness arrays). 
    Thus, to reconstruct the optimal solution we proceed by obtaining the steps which 
    define the solution $S_{\lceil \log 1/\varepsilon \rceil}[(1 + O(\varepsilon))W]$, and recursively
    find the steps which define them in the previous level. Proceeding in this way, we eventually hit the entries 
    of $\calP_0$, which correspond to the items from $\calI'$ which correspond to the solution found.
    The correctness of this procedure follows simply because we trace back the computation which led to the output value
    and store the items found in the first level. Since we can lookup the witnesses in constant time 
    and at the $i$-th level of recursion we have $2^i$ subproblems, the running time is
    $O(\sum_{i \leq \log 1/\varepsilon}2^i) = O(1/\varepsilon)$.
\end{proof}

With an analogous proof we can show the same for the strong FPTAS:
\begin{lemma}
    A solution attaining the value given by~\Cref{alg:fptas} can be found in time $\widetilde{O}(n + T(1/\varepsilon))$,
    where $T(n)$ is the time to compute \MaxConv{} on length-$n$ sequences.
\end{lemma}

\section{Hardness Results}\label{sec:hardness}
In this section we establish a reduction from \BMMaxConv{} to \UnboundedKnapsack{}. By combining this with 
our algorithms in previous sections, we establish an equivalence between \BMMaxConv{} and various \Knapsack{} problems.

\subsection{Retracing known reductions} \label{sec:hardnessretracing}

We follow the reduction from \MaxConv{} to \UnboundedKnapsack{} of~\cite{CyganMWW19} and~\cite{KunnemannPS17},
but starting from \BMMaxConv{} instead of general instances. In particular, we closely follow the structure 
of Cygan et al.'s proof. Their reduction proceeds in three steps using the following intermediate problems.
Note that we define the bounded monotone (BM) versions of their problems.

\defproblem
    {\BMMaxConvUpperBound{}}
    {Monotone non-decreasing sequences $A[0,\dots,n], B[0,\dots,n], C[0,\dots,2n]$ with entries in $[O(n)]$}
    {Decide whether for all $k \in [2n]$ we have $C[k] \geq \max_{i+j=k}A[i] + B[j]$}

\defproblem
    {\BMSuperAdditivity{}}
    {A monotone non-decreasing sequence $A[0,\dots,n]$ where each $A[i] \in [O(n)]$ for all $i \in [n]$} 
    {Decide whether for all $k \in [n]$ we have $A[k] \geq \max_{i+j=k}A[i]+A[j]$}

The first step is to reduce \BMMaxConv{} to its decision version \BMMaxConvUpperBound{}. The technique used in this 
step traces back to a reduction from $(\min,+)$-matrix product to negative weight triangle detection in graphs 
due to Vassilevska Williams and Williams~\cite{WilliamsW18}. Our proof is very similar to~\cite[Theorem 5.5]{CyganMWW19}, 
but we need some extra care to ensure that the constructed instances have bounded entries.

\begin{proposition}\label{prop:decision-to-search}
    Given an algorithm for \BMMaxConvUpperBound{} in time $T(n)$, we can \emph{find} a violated constraint
    $C[i + j] < A[i] + B[j]$ if it exists in time $O(T(n) \cdot \log n)$.
\end{proposition}
\begin{proof}
    Suppose $A, B, C$ form a NO-instance of \BMMaxConvUpperBound{} and let $k^* = i^* + j^*$ be the smallest index
    for which we have a violated constraint $C[k^*] < A[i^*] + B[j^*]$. Note that for any $k < k^*$ the prefixes 
    $A[0,\dots,k], B[0,\dots,k], C[0,\dots,k]$ form a YES instance of \BMMaxConvUpperBound{}. Thus, we can do binary 
    search over prefixes to find $k^*$.
\end{proof}

\begin{lemma}[$\BMMaxConv{} \rightarrow \BMMaxConvUpperBound{}$]\label{lem:maxconv-to-decision}
    If \styleprob{BMMaxConv UpperBound} can be solved in time $T(n)$, then 
    \BMMaxConv{} can be solved in time $\widetilde{O}(n \cdot T(\sqrt{n}))$.
\end{lemma}
\begin{proof}
    Let $A[0,\dots,n], B[0,\dots,n]$ be an input instance of \BMMaxConv{}. We will describe a procedure which given a sequence $C$ of 
    length $2n+1$, outputs for each $k \in [2n]$ whether $C[k] \geq \max_{i+j=k} A[i] + B[j]$. Given this procedure, 
    we can determine all entries of $A \oplus B$ by a simultaneous binary search using $C$ in $O(\log n)$ calls.
    Since $A, B$ are monotone non-decreasing, the sequence of guessed values $C$ will remain monotone in all iterations.

    We split the input sequences of \BMMaxConv{} as follows. Let $\Delta = O(n)$ be the largest value in $A$ and $B$. Given an interval $I \subseteq [n]$, we denote by 
    $A_I$ the contiguous subsequence of $A$ indexed by $I$.
    Among the indices $i \in [n]$, we mark every multiple of $\lceil \sqrt{n} \rceil$. We also mark the smallest index $i$ with $A[i] \ge j \cdot \lceil \sqrt{\Delta} \rceil$, for every integer $1 \le j \le \sqrt{\Delta}$. Then we split $A$ at every marked index, obtaining subsequences $A_{I_0},\ldots,A_{I_{a}}$ for $a \le n' := \lceil \sqrt{n} \rceil + \lceil \sqrt{\Delta} \rceil$. 
    We analogously construct intervals $J_0,\dots,J_b$ with $b \leq n'$ which partition $B$.
    
    Let $C[0,\dots,2n]$ be the sequence which we use to binary search the values of $A \oplus B$. Recall that our goal 
    is to determine for each $k$ whether $C[k] \geq \max_{i+j=k} A[i] + B[j]$. We will describe an iterative procedure which 
    returns a binary array $M[0,\dots,2n]$ where $M[k] = 1$ if we have determined that $C[k] < \max_{i+j=k} A[i] + B[j]$, and $M[k] = 0$
    otherwise.
    
    Initialize $M[k] = 0$ for all $k \in [2n]$. Iterate over each $(x,y) \in [a] \times [b]$. We now describe how to check if 
    $A[i] + B[j] \leq C[i + j]$ holds for all $i \in I_x$ and $j \in I_y$, or otherwise find a violated constraint using the oracle
    for \BMMaxConvUpperBound{}. Let $L := \min A_{I_x} + \min B_{J_y}$ and $U := \max A_{I_x} + \max B_{J_y}$. 
    We proceed in the following two steps:
    
    \begin{enumerate}
        \item Identify all indices $k \in I_x + I_y$ for which $C[k] < L$. Since $A[i] + B[j] \geq L$ 
        for every $i \in I_x, j \in J_y$, we conclude that $C[k] < \max_{i+j=k} A[i] + B[j]$, and thus set $M[k] = 1$
        for every such index.
        \item Add extra dummy entries $C[2n+1] := \infty$, and $M[2n+1] = 0$. Let $\mathrm{next}(k) := \min\{j \mid j \geq k, M[j] = 0\}$. 
        We construct a new sequence $C'$ by setting for every index $k \in I_x + I_y$:
        \[
            C'[k] := \min\{C[\mathrm{next}(k)], U + 1\}.
        \]
        The purpose of the dummy entries is to set $C'[k] := U + 1$ if there is no $k \leq j \leq 2n$ for which $M[j] = 0$.
        Due to the monotonicity of $C$, it follows that 
        $C'$ is also monotone. Further, it holds that $L \leq C'[k] \leq U + 1$ for every entry $k$.
        Intuitively, the purpose of $\mathrm{next}(k)$ is to ignore the entries which already have been marked 
        $M[k] = 1$ in step 1, or in previous iterations.

        Now we want to find a violating index $C'[i + j] < A_{I_x}[i] + B_{J_y}[j]$ if it exists, using 
        the \BMMaxConvUpperBound{} oracle. In order to do so, we shift the values of $A_{I_x}, B_{J_y}$ and 
        $C'$ appropriately. More precisely, let $L_A := \min A_{I_x}$ and $L_B := \min B_{J_y}$ be the lowest numbers in $A_{I_x}, B_{J_y}$ respectively.
        We substract $L_A$ from every element in $A_{I_x}$, $L_B$ from 
        every element in $B_{J_y}$ and $L_A + L_B$ from every element in $C'$. This makes all entries bounded by 
        $O(n') = O(\sqrt{n})$, and a violating index in the resulting instance is a violating index before shifting since 
        we substract the same quantity from both sides of the inequality $C'[i + j] < A_{I_x}[i] + B_{J_y}[j]$.
        Thus, we can use \Cref{prop:decision-to-search} to find a violating index $k = i + j$ if it exists,
        and if so, set $M[\mathrm{next}(k)] := 1$.
    \end{enumerate}
    
    We repeat step 2 until we find no more violating indices (recomputing the sequence $C'$ in each iteration). 
    Then we repeat the process with the next pair $(x,y)$.

    We claim that in this way, we correctly compute the array of violated indices $M$. To see correctness,
    note that step 1 is trivially correct due to the lower bound on any pair of sums. For step 2, note that if 
    we find a violating index $k=i+j$ and $\mathrm{next}(k) = k$, then setting $M[k] = 1$ is clearly correct. 
    If $\mathrm{next}(k) \neq k$, it means that we had marked the index $M[k] = 1$ in a previous iteration or in step 1. 
    Due to the monotonicity of $A, B$ and $C$ and the definition of $z := \mathrm{next}(k)$, we have that 
    \[
        C[z] = C'[k] = C'[i+j] < A[i] + B[j] \leq \max_{i'+j'=k}A[i']+B[j'] \leq \max_{i' + j' = z} A[i'] + B[j'],
    \]
    so marking $M[z] = 1$ is correct. Conversely, for any index $k=i+j$ such that 
    $C[k] < A[i]+B[j]$ at some iteration we will consider the subsequences where $i \in I_x$ and $j \in I_y$ holds.
    Since we repeat step 2 until no more violating indices are found, the index $k$ will be marked during some iteration.

    Finally we analyze the running time of the reduction. We consider $O(n)$ pairs $x,y \in [a] \times [b]$. 
    For each such pair, step 1 takes time $O(n') = O(\sqrt{n})$. We execute step 2 at least once for each pair $(x,y)$, and in total 
    once for every index that we mark as violated. Since every 
    index $k$ gets marked at most once, this amounts to $O(n)$ calls to \Cref{prop:decision-to-search}, which results in $O(n \log n)$ calls to the oracle of \BMMaxConvUpperBound{}. Each such call takes time $O(T(\sqrt{n}))$. Since we do a simultaneous binary search over $C$, the overall time of the reduction 
    is $O(n \cdot T(\sqrt{n}) \log^2 n)$.
\end{proof}

The next step is to reduce \BMMaxConvUpperBound{} to \BMSuperAdditivity{}. The proof is exactly the same as 
\cite[Theorem 5.4]{CyganMWW19}; we include it for completeness.

\begin{lemma}[$\BMMaxConvUpperBound{} \rightarrow \BMSuperAdditivity{}$]
    If \BMSuperAdditivity{} can be solved in time $T(n)$, then \BMMaxConvUpperBound{} can be solved in time $O(T(n))$.
\end{lemma}
\begin{proof}
    Let $\Delta = O(n)$ be the maximum number in the input sequences $A,B,C$ of \BMMaxConvUpperBound{}. We construct  
    an equivalent instance $E[0,\dots,4n+3]$ of \BMSuperAdditivity{} as follows: for each $i \in [n]$ set $E[i] := 0$,
    $E[n+1 + i] := \Delta + A[i]$, $E[2n+2 + i] := 4\Delta + B[i]$ and $E[3n+3 + i] := 5\Delta + C[i]$. Note that 
    $|E| = O(n)$ and all values are bounded by $O(\Delta) = O(n)$.

    If there are indices $i,j \in [n]$ such that $A[i] + B[j] > C[i + j]$, then $E[n + 1 + i] + E[2n + 2 + j] > E[3n + 3 + i + j]$,
    so $E$ is a NO instance. Otherwise, let $i \leq j$ with $i + j \leq 4n+3$ be any pair of indices. 
    If $i \leq n$, then since $E[i] = 0$ we have $E[i] + E[j] \leq E[i+j]$. So assume $i > n$. If $j \in [n+1, 2n+1]$ then
    $E[i] + E[j] \leq 4\Delta \leq E[i + j]$. Hence, since $i \leq j$ and $i + j \leq 4n+3$ we have that 
    $i \in [n+1, 2n+1]$ and $j \in [2n+2, 3n+2]$. For these ranges, the super-additivity of $E$ corresponds
    exactly to the \BMMaxConvUpperBound{} condition of $A,B,C$.
\end{proof}

Finally, we reduce \BMSuperAdditivity{} to \UnboundedKnapsack{}. The idea is essentially the same 
as~\cite[Theorem 5.3]{CyganMWW19}, but in their proof they construct instances of \UnboundedKnapsack{} where the maximum profit could be as large as $p_{\max} = O(n^2)$; we improve this to $p_{\max} = O(n)$. (More precisely,
the variable $D$ is defined as $D := \sum_i A[i]$ in their proof, while we show that it is enough to set $D = O(A[n])$.)

\begin{lemma}[$\BMSuperAdditivity{} \rightarrow \UnboundedKnapsack{}$]
    There is a linear-time reduction from \BMSuperAdditivity{} on sequences of length $n$ to an instance of 
    \UnboundedKnapsack{} on $n$ items with $W, \OPT = O(n)$.
\end{lemma}
\begin{proof}
    Let $A[0,\dots,n]$ be an instance of \BMSuperAdditivity{}. We construct an equivalent instance of 
    \UnboundedKnapsack{} as follows. Set $D := 5 \cdot A[n]$ and $W := 2n + 1$. For every $i \in [n]$ we create 
    a \emph{light item} $(A[i], i)$ and a \emph{heavy item} $(D - A[i], W - i)$, where the first entry of each item
    is its profit and the second is its weight. We set the weight constraint to $W$.

    Suppose $A[0,\dots,n]$ is a NO instance. Then, there exist $i, j \in [n]$ such that $A[i] + A[j] > A[i + j]$.
    Hence, $(A[i], i), (A[j], j)$ and $(D - A[i + j], W - i - j)$ form a feasible solution for \UnboundedKnapsack{} with weight $W$ 
    and total profit $A[i] + A[j] + D - A[i+j] \geq D + 1$.

    On the other hand, suppose we start from a YES instance. Consider any feasible solution to the 
    \UnboundedKnapsack{} instance. We make two observations. First, if the solution contains any two light 
    items $(A[i], i)$ and $(A[j], j)$ with $i + j \leq n$, then since $A$ is a YES instance, we can replace both 
    items by $(A[i + j], i + j)$ and the profit is at least as good. The second observation is that any feasible 
    solution can contain at most one heavy item, since every such item has weight at least $n + 1$ and the weight
    constraint is $W = 2n + 1$. Now, suppose that $\OPT$ contains one heavy item of the form $(D - A[k], W - k)$. By 
    the first observation, the most profitable way of packing the remaining capacity is to include the item 
    $(A[k], k)$. Thus, in this case, we have that the value of $\OPT$ is $D - A[k] + A[k] = D$. If $\OPT$ does not
    contain any heavy item, note that by the first observation it consists of at most 4 light items. Therefore, its 
    value is at most $4 A[n] < D = 5 A[n]$. 
\end{proof}

Putting together the previous lemmas, we obtain the following theorem.

\begin{theorem}\label{thm:reduction-maxconv-to-knapsack}
    If \UnboundedKnapsack{} on instances with $n$ items with $W, \OPT = O(n)$
    can be solved in time $T(n)$, then \BMMaxConv{} on sequences of length $n$ can be solved in time 
    $\widetilde{O}(n\cdot T(\sqrt{n}))$.
\end{theorem}

We will also need the following reduction from \UnboundedKnapsack{} to \Knapsack{} from~\cite[Theorem 5.1]{CyganMWW19}; we include the proof for completeness.

\begin{lemma}[$\UnboundedKnapsack{} \rightarrow \Knapsack{}$]\label{lem:0-1knapsack-to-unbounded}
    If \Knapsack{} on instances with $n$ items and $W, \OPT = O(n)$ can be solved in 
    time $T(n)$, then \UnboundedKnapsack{} on instances with $n$ items and $W, \OPT = O(n)$ can be solved in time $\widetilde{O}(T(n))$.
\end{lemma}
\begin{proof}
    Let $\calI = \{(p_i, w_i)\}_{i \in [n]}$ with capacity $W$ be an instance of \UnboundedKnapsack{}.
    We construct an equivalent instance of \Knapsack{} with the item set 
    \[
        \calI' := \{(2^j p_i, 2^j w_i) \mid (p_i, w_i) \in \calI, \; 0 \leq j \leq \log W \},
    \]
    and the same capacity $W$. Let  $x \in \N^n$ be a solution of the \UnboundedKnapsack{} instance~$\cal I$, where 
    $x_i$ is the multiplicity of item $(p_i, w_i)$. We can construct an equivalent solution $x' \in \{0,1\}^{|\calI'|}$ 
    of the \Knapsack{} instance $\calI'$ by expressing each $x_i \leq W$ in binary and adding to $x'$ the items from $\calI'$ corresponding 
    to the non-zero coefficients. In this way, $w_{\calI}(x) = w_{\calI'}(x')$ and $p_{\calI}(x) = p_{\calI'}(x')$. 
    It is easy to see that this mapping can be inverted, which establishes the equivalence between the instances. 
    
    Note that this reduction does not change $W$ and $\OPT$, and $n$ increases to $n \log W = O(n \log n)$. Thus for \UnboundedKnapsack{} we obtain running time $O(T(n \log n)) = \widetilde{O}(T(n))$. Here we used the niceness assumption on $T(n)$.
\end{proof}

\subsection{Consequences}

In this section we combine the reductions from~\Cref{sec:hardnessretracing}, as well as our algorithms from \Cref{sec:approx}, \Cref{sec:exact-unboundedknapsack}, and \Cref{sec:exactzerooneknapsack}, to prove~\Cref{thm:equivalence}.

We start by showing how a weak approximation for \UnboundedKnapsack{} gives an exact algorithm for \BMMaxConv{}.

\begin{lemma}[$\BMMaxConv{} \rightarrow \text{Approximate }\UnboundedKnapsack{}$]\label{lem:reduction-maxconv-to-approxknapsack}
    If \UnboundedKnapsack{} has a weak approximation scheme running in time $T(n, \varepsilon)$,
    then \BMMaxConv{} can be solved in time $O(n \cdot T(\sqrt{n}, 1/\sqrt{n}))$.
\end{lemma}
\begin{proof}
    Note that a weak approximation scheme for \UnboundedKnapsack{} allows us to solve instances \emph{exactly} 
    by setting $\varepsilon := \Theta(1/(W + \OPT))$. Since the reduction from~\Cref{thm:reduction-maxconv-to-knapsack}
    produces $O(n)$ instances of \UnboundedKnapsack{} with $\sqrt{n}$ items and $W, \OPT = O(\sqrt{n})$, we obtain an algorithm 
    for \BMMaxConv{} by setting $\varepsilon = \Theta(1/\sqrt{n})$. The overall running time becomes $O(n \cdot T(\sqrt{n}, 1/\sqrt{n}))$.
\end{proof}

Finally, we put the pieces together to prove~\Cref{thm:equivalence}, which we restate for convenience:

\equivalence*

\begin{proof}
    In the following, we write $A \rightarrow B$ to denote a reduction from problem $A$ to problem $B$. For every reduction stated in the following
    list, we obtain the stronger guarantee that if $B$ can be solved in time 
    $\widetilde{O}(n^{2-\delta})$, then $A$ can be solved in time $\widetilde{O}(n^{2-\delta})$, i.e. without any loss in the exponent:
    \begin{itemize}
        \item $(2) \rightarrow (1)$: Follows from~\Cref{thm:reduction-unboundedknapsack-to-conv}.
        \item $(2) \rightarrow (3)$: Follows from~\Cref{lem:0-1knapsack-to-unbounded}.
        \item $(3) \rightarrow (1)$: Follows from~\Cref{lem:reduction-0-1knapsack-to-maxconv}.
        \item $(4) \rightarrow (1)$: Follows from~\Cref{thm:reduction-weakfptas-to-minconv}.
    \end{itemize}

    For the remaining two reductions below, we obtain the weaker guarantee as stated in the theorem. Namely, if $B$ can be solved 
    in time $\widetilde{O}(n^{2-\delta})$, then $A$ can be solved in time $\widetilde{O}(n^{2-\delta/2})$:

    \begin{itemize}
        \item $(1) \rightarrow (2)$: Follows from~\Cref{thm:reduction-maxconv-to-knapsack}.
        \item $(1) \rightarrow (4)$: Follows from~\Cref{lem:reduction-maxconv-to-approxknapsack}.
    \end{itemize}

    Finally, note that for any pair of problems $A$ and $B$ in the theorem statement, we can reduce $A$ to $B$ by chaining 
    the reductions written in the previous two lists and use at most one reduction which reduces the saving to $\delta/2$.
    This guarantees that if $B$ can be solved in time $\widetilde{O}(n^{2-\delta})$, then $A$ can be solved in time 
    $\widetilde{O}(n^{2-\delta/2})$, as desired.
\end{proof}

\bibliographystyle{alpha}
\bibliography{refs}

\appendix
\section{Missing Proofs from \Cref{sec:preliminaries}}\label{sec:proofs-preliminaries}

\begin{proof}[Proof of~\Cref{prop:maxconv-range}]
    By shifting the indices, we can assume that $A[I]$ and $B[J]$ are sequences $A'[0,\dots,|I|-1]$ and 
    $B'[0,\dots,|J|-1]$. Compute $C' := A' \oplus B'$ in time $T(|A| + |B|)$. By shifting the indices 
    back, we can infer the values of the entries $C[I + J] = A[I] \oplus B[J]$. Thus, we can simply read off the 
    entries in $C[K]$ from the array $C'$.
\end{proof}

\paragraph*{Equivalence between variants of \BMMaxConv{}}\label{sec:appendix-equivalence}

As stated in~\Cref{sec:preliminaries}, the algorithm of Chi, Duan, Xie and Zhang~\cite{ChiDXZ22} computes the 
$(\min,+)$-convolution of sequences which are monotone increasing and have values in $[O(n)]$. The following proposition
shows that this is equivalent to solving \MaxConv{} on monotone non-decreasing sequences with values 
in $[O(n)] \cup \{-\infty\}$, which justifies~\Cref{thm:convolution-bounded}.

\begin{proposition}
    \MaxConv{} on monotone non-decreasing sequences of length $n$ and values in $[O(n)] \cup \{-\infty\}$ is 
    equivalent to \MinConv{} on monotone increasing sequences of length $n$ and values in $[O(n)]$, in the 
    sense that if one can be solved in time $T(n)$ then the other can be solved in time $O(T(n))$.
\end{proposition}
\begin{proof}
    We first describe how to reduce \MaxConv{} on monotone non-decreasing sequences and values in $[O(n)] \cup \{-\infty\}$
    to \MinConv{} on monotone increasing sequences and values in $[O(n)]$ via a simple chain of reductions:
    \begin{itemize}
    \item \emph{Removing $-\infty$:}
    Let $A[0,\dots,n], B[0,\dots,n]$ be an instance of \MaxConv{} where $A, B$ are monotone non-decreasing and 
    $A[i], B[i] \in [O(n)] \cup \{-\infty\}$. We start by reducing it to an equivalent instance of \MaxConv{}  
    on monotone non-decreasing sequences and values in $[O(n)]$ (i.e. we remove the $-\infty$ entries).
    Let $\Delta$ be the maximum entry of $A$ and $B$. Construct a new sequence $A'[0,\dots,n]$
    where $A'[i] := 0$ if $A[i] = -\infty$, and $A'[i] := A[i] + 2\Delta$ otherwise. Construct $B'[0,\dots,n]$
    from $B$ in the same way. Note that $A'$ and $B'$ are monotone non-decreasing and have values in $[O(n)]$.
    Moreover, we can infer the values of any entry $(A \oplus B)[k]$ from $C'$: if $C'[k] \leq 3\Delta$
    then $(A \oplus B)[k] = -\infty$ and otherwise $(A \oplus B)[k] = C'[k] - 4\Delta$.

    \item \emph{Reducing to \MaxConv{} on non-increasing sequences:}
    Now we reduce an instance $A[0,\dots,n]$, $B[0,\dots,n]$ of \MaxConv{} on monotone non-decreasing
    sequences and values in $[O(n)]$ to an instance of \MinConv{} on monotone non-increasing sequences and values in $[O(n)]$.
    Let $\Delta$ be the maximum entry of $A$ and $B$. Construct two new sequences $A'$ and $B'$ by setting $A'[i] := \Delta - A[i]$ and $B'[i] := \Delta - B[i]$. Then 
    $A'$ and $B'$ are monotone non-increasing and given their $(\min,+)$-convolution we can easily infer the 
    $(\max,+)$-convolution of $A$ and $B$. 

    \item \emph{Reducing to \MinConv{} on increasing sequences:}    
    Next, we reduce an instance $A[0,\dots,n]$, $B[0,\dots,n]$ of \MinConv{} on monotone non-increasing sequences and 
    values in $[O(n)]$ to an instance of \MinConv{} on increasing sequences and values in $[O(n)]$. Construct two new sequences 
    $A'$ and $B'$ by reversing and adding a linear function to $A$ and $B$, i.e., set $A'[i] := A[n - i] + i$ and
    $B'[i] := B[n - i] + i$ for every $i \in [n]$. Note that $A'$ and $B'$ are monotone increasing sequences, 
    and given their $(\min,+)$-convolution we can infer the $(\min,+)$-convolution of $A$ and $B$.
    \end{itemize}

    Combining the reductions above, we conclude that \MaxConv{} on monotone non-decreasing sequences with values in $[O(n)] \cup \{-\infty\}$
    can be reduced in linear time to \MinConv{} on monotone increasing sequences with values in $[O(n)]$.
    To show the reduction in the other direction, we can apply the same ideas:
    we first negate the entries and shift them to make them non-negative, then reverse the resulting sequences. We omit the details.
\end{proof}

\section{Witnesses for \BMMaxConv{}}\label{sec:witnesses}

In this section we give the proof of~\Cref{lem:witness-finding}.
For the remainder of this section, fix an instance $A[0,\dots,n], B[0,\dots,n]$ of \BMMaxConv{} 
and let $C := A \oplus B$. To compute the witness array $M[0,\dots,2n]$, we first show that if an entry 
$C[k]$ has a unique witness then we can easily find it. Then we reduce to the unique witness case with 
randomization. This approach is well known~\cite{Seidel95,AlonGMN92}, but some extra care is needed to 
ensure that the instances remain monotone and bounded. The standard derandomization for this
approach~\cite{AlonN96} also works in our setting.

\begin{lemma}\label{lem:unique-witness}
    If \BMMaxConv{} on length-$n$ sequences can be computed in time $T(n)$, then in time $\widetilde{O}(T(n))$ we 
    can compute an array $U[0,\dots,2n]$ such that for every $k \in [2n]$ if 
    $C[k]$ has a unique witness $M[k]$ then $U[k] = M[k]$. The output $U[k]$ is undefined otherwise.
\end{lemma}
\begin{proof}
    We will describe how to compute the unique witnesses bit by bit. For each bit position 
    $b \in [\lceil \log n \rceil]$ let $i_b \in \{0,1\}$ be the $b$-th bit of $i$.
    Construct sequences $A_b, B_b$ defined as $A_b[i] := 2 A[i] + i + i_b$ and $B_b[i] := 2 B[i] + i$ for $i \in [n]$. 
    Note that $A_b, B_b$ are still monotone non-decreasing and have entries bounded by $O(n)$.
    Compute $C_b := A_b \oplus B_b$. For each $k \in [2n]$, we set the $b$-th bit of 
    $U[k]$ to  $C_b[k] \bmod 2$. It is not hard to see that for those entries $C[k]$ which have unique witnesses $U[k]$,
    this procedure indeed gives the $b$-th bit of $U[k]$. Indeed, note that because we double every entry in $A_b, B_b$ 
    and add a linear function, adding $i_b$ does not 
    change the maximizer. Therefore, if $C[k]$ has a unique witness then $C'[k]$ has a unique witness whose $b$-th bit can be read 
    from the least significant bit of $C'[k]$. Thus, by repeating this over all bit positions $b \in [\lceil \log n \rceil]$, we compute 
    the entire array of unique witnesses $U$ using $O(\log n)$ invocations to \BMMaxConv{}, as desired.
\end{proof}

Now we prove the main lemma of this section which we restate for convenience.
\witnesses*
\begin{proof}[Proof of~\Cref{lem:witness-finding}]
    Fix a set $S \subseteq [n]$. We say that an entry $C[k]$ gets \emph{isolated} by $S$ if the number of witnesses 
    of $C[k]$ in $S$ is exactly one. We will now describe how to find the witnesses of all entries isolated by $S$ 
    (the idea and argument is similar as in the proof of~\Cref{lem:unique-witness}).
    Construct sequences $A', B'$ where for each $i \in [n]$ we set 
    \[
        A'[i] := \begin{cases}
                    2 A[i] + i + 1 &\text{if } i \in S \\
                    2 A[i] + i &\text{otherwise} 
        \end{cases}
    \]
    and $B'[i] := 2B[i] + i$. Note that these sequences are monotone non-decreasing and have entries bounded by $O(n)$.
    Let $C' := A' \oplus B'$. We claim that if an entry $C[k]$ is isolated by $S$, then $C'[k]$ has a unique witness.
    To see this, note that if no witness of $C[k]$ gets included in $S$, then we have that $C'[k] = 2C[k] + k$. If at least
    one witness gets included in $S$, then $C'[k] = 2C[k] + k + 1$. In particular, if a witness gets isolated then $C'[k]$
    will have a unique witness, as claimed. Thus, by~\Cref{lem:unique-witness} we can compute in time 
    $\widetilde{O}(T(n))$ an array $U[0,\dots,2n]$
    which contains the witnesses of all entries that are isolated  by $S$. Note that some entries of $U$
    might be undefined, but we can simply check in time $O(n)$ which entries of $U$ are true witnesses, by iterating over $k \in [2n]$ and checking whether the equality $C[k] = A[U[k]] + B[k - U[k]]$ holds.

    Now we show how to select appropriate sets $S$. Fix an entry $C[k]$ and denote by $R \in [n]$ its number of witnesses.
    We sample $S \subseteq [n]$ by including each element $i \in [n]$ independently with probability $p := 2^{-\alpha}$ where 
    $\alpha \in \N$ is chosen such that $2^{\alpha-2} \leq R \leq 2^{\alpha-1}$. 
    Let $X$ be the random variable counting the number of witnesses of 
    $C[k]$ that get sampled in $S$. By keeping the first two terms in the inclusion-exclusion formula we have that 
    $\Pr[X \geq 1] \geq p \cdot R - {R \choose 2}p^2$, and by a union bound $\Pr[X \geq 2] \leq {R \choose 2}p^2$. Thus, 
    \[
        \Pr[X = 1] = \Pr[X \geq 1] - \Pr[X \geq 2] \geq p\cdot R(1 - p\cdot R) \geq 1/8
    \]
    where the last inequality holds because $1/8 \leq p\cdot R \leq 1/4$ due to the choice of $\alpha$. In particular, 
    $S$ isolates $C[k]$ with probability at least $1/8$.

    We now put the pieces together. Iterate over the $O(\log n)$ possible values for $\alpha$. Sample a set $S$
    and find all witnesses of entries isolated by $S$ as described earlier in time $\widetilde{O}(T(n))$.
    As we argued above, if $C[k]$ has $R$ witnesses and $2^{\alpha-2} \leq R \leq 2^{\alpha - 1}$, then $C[k]$ gets isolated 
    with constant probability. Thus, by repeating this step with the same $\alpha$ for $O(\log n)$ freshly 
    sampled sets $S$ we find a witness for all such entries $C[k]$ with probability at least $1 - 1/\poly(n)$. Combining
    the results across iterations we obtain the array of witnesses $M[0,\dots,2n]$ in time $\widetilde{O}(T(n))$, as desired.

    Finally, we note that this procedure can be derandomized with standard techniques~\cite{AlonN96}.
\end{proof}

\end{document}